\documentclass[leqno,11pt]{article}
\usepackage{graphicx}
\usepackage{epsfig}
\usepackage[english]{babel}
\usepackage{amsfonts,amsmath,enumerate}
\usepackage{amscd}
\usepackage{amsmath,amsfonts,amssymb,mathrsfs}

\usepackage{latexsym}
\usepackage{stmaryrd}
\usepackage{amssymb}
\usepackage{theorem}
\usepackage{color}

\newlength{\oldparindent}
\def\lbr{[\![}
\def\rbr{]\!]}

\title{\Large \bf  Portfolio optimization with insider's initial information and  counterparty risk   \thanks{ 
This research is part of a project of Europlace Institute of Finance.  We thank Laurent Denis, Nicole El Karoui, Monique Jeanblanc, Huy\^en Pham, Abass Sagna, Nizar Touzi and Lioudmila Vostrikova for discussions.
} }
\author{\textsc{Caroline HILLAIRET\thanks{CMAP Ecole Polytechnique, Email: caroline.hillairet@polytechnique.edu.  Financial support by Chair {\it Financial Risks} of the {\it Risk
Foundation},  Chair {\it
Derivatives of the Future} sponsored by the {F\'ed\'eration Bancaire
Fran\c{c}aise},  Chair {\it Finance and Sustainable Development}
sponsored by EDF and Calyon. 
}\quad  Ying JIAO\thanks{LPMA Universit\'e Paris Diderot  and Peking University,  Email: jiao@math.univ-paris-diderot.fr. Financial support by Alma Recherche.  
}}\\
}

\date{\today}

\newtheorem{definition}{Definition }[section]
{Proposition }
{Lemma }%
{Theorem }
{Corollary }
{Remark }
\newtheorem{hypothese}[definition]%
{Assumption}

\theorembodyfont{\upshape}
{Example}

\newcommand{\esp}{\mathbb E} 
\newcommand{\indic}{1\!\!\!1} 
\newcommand{\proba}{\mathbb{P}}

\newcommand{\bF}{\mathbb{F}}

\newcommand{\R}{\mathbb{R}}
\newcommand{\D}{\mathcal{D}}
\newcommand{\F}{\mathcal{F}}

\newcommand{\G}{\mathcal{G}}
\newcommand{\bG}{\mathbb{G}}

\newcommand{\cA}{\mathcal{A}}
\newcommand{\cP}{\mathcal{P}}

\def\esssup_#1{\underset{#1}{\mathrm{ess\,sup\, }}}
\def\argmax_#1{\underset{#1}{\mathrm{arg\,max\, }}}
\def\essinf_#1{\underset{#1}{\mathrm{ess\,inf\, }}}

\theorembodyfont{\upshape\it}
\newtheorem{Thm}{\bf Theorem}[section]
\newtheorem{Pro}[Thm]{\bf Proposition}
\newtheorem{Lem}[Thm]{\bf Lemma}

\theorembodyfont{\upshape}

\newtheorem{Rem}[Thm]{Remark}

\newcommand {\proof} {\noindent {\sc Proof:  }}

\newcommand {\finproof} {\hfill $\Box$ \vskip 5 pt }

\begin{document}

\maketitle

\numberwithin{equation}{section}


\begin{abstract}

We study the gain of an insider having private information which concerns the default risk of a counterparty.  More precisely,  the default time $\tau$ is modelled as the first time  a stochastic process hits a random barrier $L$. The insider knows this barrier (as it can be the case for example for the manager of the counterparty), whereas standard investors only observe its value at the default time.  All investors aim to maximize the expected utility from terminal wealth, on  a financial market where the risky asset price  is exposed to a sudden loss at the default time of the counterparty.  In this framework, the insider's information is modelled by using an initial enlargement of filtration and $\tau$ is a stopping time with respect to this enlarged filtration. We prove that the regulator must impose short selling constraints for the insider, in order to exclude the value process to reach infinity.  We then solve the optimization problem and we study the gain of the insider,   theoretically and numerically. In general, the insider achieves a larger value of expected utility than  the standard investor. But in extreme situations for the default and loss risks,  a standard investor may in average  outperform the insider, by taking advantage of an aggressive short selling position which is not allowed for the insider, but at the risk of big losses if the default finally occurs after the maturity. 


\end{abstract}

\vspace*{0.5cm}



\newpage

\section{Introduction}
The insider's optimal investment is a classical problem where an investor possessing some extra flow of information aims to maximize the expected utility on the final value of her portfolio. As the insider has more information, she has access to a larger set of available trading strategies, leading to a higher expected utility from terminal wealth. In the literature, an interesting question has been studied: what is the cost of the extra information? From an indifference point of view, we search for the value at which the investor accepts to buy the information at the initial time, that is, the amount of money she is ready to pay such that this cost is offset by the increase of the maximal expected utility. This is the approach adopted by Amendinger et al. \cite{ABS}, where the authors study the value of an initial information in the setting of a complete default free market. The extra information they consider is a terminal information distorted by an independent noise, for example, a noisy signal of a functional of the final value of the assets. We adopt a more direct manner: we are interested in the gain of the insider from her investment strategy on the portfolio compared to other investors not having access to the extra information. The originality of  our paper is to study this problem in the context of credit risks: the insider's information concerns the default risk of a counterparty firm. 

During the financial crisis, the counterparty default has become an important source of risk we need to take into account. Jiao and Pham \cite{jp} have considered an optimal investment problem where the risky asset in the portfolio is subjected to the default risk of a counterparty firm and its value may suffer a sudden loss at the counterparty default time $\tau$. This paper is a good benchmark in our study in order to quantify the value of the extra information.  The accessible information for a standard investor is described as in the classical credit risk modelling by Bielecki and Rutkowski \cite{BR},   using the progressive enlargement of a reference ``default-free'' filtration $\bF=(\F_t)_{t\geq 0}$ by the default $\tau$. To analyze the impact of default, the default density framework developed in El Karoui et al. \cite{ejj} has been adopted. 

This current paper concentrates on an insider in comparison with a standard investor. Both agents can invest in the same risk-free asset and  risky one and they observe the same market price for each asset. However, the insider possesses more information on the risky asset  since it is influenced by the counterparty default on which  the insider has additional knowledge.  Due to the extra information, the insider may gain larger profit.  The insider's information is modelled by using an initial enlargement of filtration as in \cite{ABS} and in Grorud and Pontier \cite{Pon}.  More precisely, in the credit risk context, we model the default time $\tau$ as the first time that a stochastic process hits a random barrier $L$. The insider knows the barrier from the initial time  and the other investors only see its value at the default time. This extra information is called the insider's information, or the full information in Hillairet and Jiao \cite{nous}.

We shall consider the insider's optimization problem in parallel with the one studied in \cite{jp}. The canonical decomposition of processes adapted to the enlarged filtration induces to specify  the investment strategies  on the two sets: before-default one $\{t<\tau\}$ and  after-default one $\{t\geq\tau\}$, which is a similar point to \cite{jp}. However, due to the extra knowledge on the default barrier $L$, the insider's strategy depends on $L$ before the counterparty default, which is not the case for the standard investor. If the default occurs, the insider's strategy will depend on the default time $\tau$. From the methodology point of view, the main difference here is that for the insider, the default time is modelled as in the classical structural approach model since the random barrier $L$ is known, so that $\tau$ becomes a stopping time w.r.t. the reference filtration $\bF$. Therefore, the default density hypothesis, which is crucial in \cite{jp}, fails to hold for the insider and we can no longer adopt the conditional density approach in this situation. We apply the theory of initial enlargement of filtration, assuming that the conditional law of $L$ given $\F_t$ is equivalent to the law of $L$. The corresponding Radon Nikodym derivative process, $(p_t(.),t\geq 0)$ will play a  key role in our methodology.

The main observation of our study is that, if the short-selling is not regulated, then the insider can obtain unbounded terminal wealth. This justifies the necessity to consider the optimization problem with portfolio strategies where the short-selling is limited to a given level. We decompose the optimization problem as an after-default problem and a global before-default one, that we solve respectively by using the dual and the dynamic programming methods. 
 To make comparison with the standard investor in \cite{jp}, we choose to consider CRRA utility function.     

The paper is organized as follows. In section 2, we introduce the model for the counterparty default, and we define and compare the informational structure of an insider with respect to a standard investor. In Section 3 we present the insider's investment problem and we decompose it into an after-default and global before-default ones, using the Radon Nikodym derivative process.  We also prove the necessity of imposing short selling constraint for the insider to exclude the value process to reach infinity. In Section 4  we solve the two optimization problems: the after-default  one through duality methods in a default free complete market, and the global before-default one through dynamic programming approach. We perform theorical comparison of the value  process of the after-default optimization problem,  and for the global before-default optimization problem, the comparison is done in 
Section 5 through  numerical illustrations.

\section{Counterparty default model and information}

We first introduce the model for the counterparty default, which is a general and standard model in the credit risk analysis. Let us fix a probability space $(\Omega, \mathfrak A,\proba)$ equipped with   a  reference  filtration $\bF=(\F_t)_{t\geq 0}$ satisfying the usual conditions, which represents the ``default-free'' information.
 Let $\tau$ be a positive random time denoting the default time of the counterparty, which is not necessarily an $\bF$-stopping time.

\noindent $\bullet$ {\it The default model} 

We consider the default risk of the counterparty in a general barrier model.
Let $(\lambda_t,t\geq 0)$ be  a positive $\bF$-adapted process  representing the default intensity process of the  counterparty. {Denote by $\Lambda_t=\int_0^t\lambda_sds$. It is an increasing process.} We model the default time as the first passage time of the  process $\Lambda$ to a {positive} random barrier $L$, i.e.
\begin{equation}\tau=\inf\{t\geq 0: \Lambda_t\geq L\}
\end{equation}where the default threshold $L$ is a positive $\mathfrak A$-measurable random variable.
In the particular but widely used case of Cox process model, $L$ is independent of $\F_\infty$ and follows the uni-exponential law. In the case where $L$ is constant or deterministic, $\tau$ is an $\bF$-stopping time as in the classical structural default models. 

\noindent $\bullet$ {\it Information of the insider}

We suppose that,  besides the information on the ``default-free'' market, the insider has complete information on $L$:  this is the case for example of the counterparty firm's managers who determine the default threshold.   
This full information is modelled as the initial enlargement of the filtration $\bF$ by $L$ and denoted by $\bG^M=(\G_t^M)_{t\geq 0}, \,\G_t^M=\F_t\vee\sigma(L)$.  
Without loss of generality, we {assume} that all the filtrations we deal with in the following satisfy the usual conditions.

We suppose that a standard investor on the market observes  whether the default has occurred or not and if so, the default time $\tau$, together with the information  contained in the  filtration $\bF$. Mathematically, this information is represented by the progressive enlargement of filtration $\bF$ by $\tau$, or more precisely, by the filtration $\bG=(\G_t)_{t\geq 0}$ where
$\G_t=\F_t\vee\D_t, \,\D_t=\sigma(\indic_{\tau\leq s},s\leq t)$. This is the standard credit risk modeling for an market investor as in \cite{BR}.

The investor's information is included in the insider's information flow. We have $\G_t\subseteq\G_t^M$ for any $t\geq 0$. In fact, before the default $\tau$, i.e., on the set $\{t<\tau\}$, the insider has additional information on $L$, so her information $\G_t^M$ is {in general} strictly larger than $\G_t$. After the default occurs, both of them observe the default event and subsequently the value of  $L$ so that they have equal information flow.  

We recall  the canonical decomposition of $\bG^M$-adapted (respectively $\bG^M$-predictable) processes (see Jeulin  \cite{Jeu} Lemma 3.13 and 4.4).
\begin{Lem}\label{Lem:decomp}\begin{enumerate}\item For  $t\geq 0$, any $\G_t^M$-measurable random variable can be written in the form $Y_t=\indic_{\tau>t}Y_t^0(L)+\indic_{\tau\leq t}Y_t^1(\tau)$ where {$Y_t^0(\cdot)$ and $Y_t^1(\cdot)$} are $\F_t\otimes\mathcal B(\R_+)$-measurable.
\item   Any $\bG^M$-adapted process $Y$ admits the decomposition form $Y_t=\indic_{\tau > t}Y_t^0(L)+\indic_{\tau \leq t}Y_t^1(\tau)$ where {$Y^0(\cdot)$ and $Y^1(\cdot)$} are $\bF \otimes\mathcal B(\R_+)$-adapted\footnote{{Namely for any $t\geq 0$, the function $Y^i_t(\cdot)$ is $\mathcal F_t\otimes\mathcal B(\mathbb R_+)$-measurable.}}. 
\item  Any $\bG^M$-predictable process $Y$ admits the decomposition form $Y_t=\indic_{\tau\geq t}Y_t^0(L)+\indic_{\tau< t}Y_t^1(\tau)$ where {$Y^0(\cdot)$ and $Y^1(\cdot)$} are $\cP(\bF)\otimes\mathcal B(\R_+)$-measurable, {$\cP(\bF)$ being the predictable $\sigma$-algebra associated with the filtration $\bF$}.
\end{enumerate}\end{Lem}

\begin{Rem}To compare with the case of a standard investor, we recall that any $\G_t$-measurable random variable $Z_t$ can be written as $Z_t=\indic_{\tau>t} Z_t^0+\indic_{\tau\leq t}Z_t^1(\tau)$ where $Z_t^0$ and {$Z_t^1(\cdot)$} are respectively $\F_t$-{measurable} and $\F_t\otimes\mathcal B(\R_+)$-measurable.
\end{Rem}

\section {Insider's optimization problem}

\subsection{Portfolio investment strategy and wealth process}
From now on, a finite horizon $T$ is fixed and all investment strategies take place from time $0$ to time $T$.
The insider has access to the same financial market as the standard investor, more precisely, she can invest in two types of financial assets. The first one is a risk-free bond with strictly positive values. We choose it as the num\'eraire and assume, without loss of generality that the value of this bond equal to $1$. The other asset is a risky one which is affected by default risk of  the counterparty firm on which the insider has extra information. 

The price  of this risky asset is observable by all investors on market at any time $t\in [0,T]$. Since it is subject to the counterparty default risk, the price process is modelled by a $\bG$-adapted process $S$, which admits the  decomposition form
\begin{equation}
S_t=S_t^{0}\indic_{t<\tau}+S_t^{1}(\tau)\indic_{t\geq\tau}, \quad 0 \leq t \leq T
\end{equation}where $S^{0}$ is $\bF$-adapted and {$S^{1}(\cdot)$} is $\bF\otimes\mathcal B(\R_+)$-adapted. We suppose that the asset suffers a contagious loss at the default time of the counterparty, that is, \[S_\theta^{1}(\theta)=S_{\theta-}^{0}(1-\gamma_\theta),\]  and that $0<\gamma_\theta< 1$ for any $\theta\geq 0$ so that the asset price remains strictly positive. The process $\gamma$ is $\bF$-adapted and  represents the proportional loss at default. 

We consider the trading strategy of the insider, who chooses to adjust the portfolio of assets according to information accessibility. Therefore, the investment strategy process is characterized by a $\bG^M$-predictable process $\pi$ which represents the proportion of wealth invested in the risky asset and is of the form
\[\pi_t=\indic_{t\leq\tau}\pi_t^{0}(L)+\indic_{t>\tau}\pi_t^1(\tau),\] where $\pi^{0}(\cdot)$ and $\pi^{1}(\cdot) $ are $\cP(\bF)\otimes\mathcal B(\R_+)$-measurable processes.  Starting from an initial wealth $X_0\in \R_+$, the total wealth of the insider's portfolio  is then a $\bG^M$-adapted process given by 
\begin{equation}\label{wealth decomposition}X_t=\indic_{t<\tau}X^{0}_t(L)+\indic_{t\geq\tau}X^{1}_t(\tau)\end{equation}
where  the before-default wealth process satisfies the self-financing equation
\begin{equation}
dX_t^{0}(L)=X_t^{0}(L)\pi_t^{0}(L)\frac{dS_t^{0}}{S_t^{0}},\quad  0 \leq t \leq T
\end{equation}
and after the default $\tau$, the wealth process has a change of regime in its dynamics and satisfies
\begin{equation}dX_t^{1}(\tau)=X_t^{1}(\tau)\pi_t^{1}(\tau)\frac{dS_t^{1}(\tau)}{S_t^{1}(\tau)},\quad  t \in \lbr \tau, T \rbr.
\end{equation}
At the default time, the wealth has a jump in its value. 
 Therefore, at time $\tau$, the initial value of the after-default wealth process is \begin{equation}\label{X1 initial}X_{\tau}^{1}(\tau)=X^{0}_{\tau-}(L)\big(1-\pi_\tau^{0}(L)\gamma_\tau\big).\end{equation} 
 We suppose that
$\pi_\tau^{0}(L)\gamma_\tau<1$, so that the wealth remains  strictly positive after the jump due to the counterparty default.


We consider  the following dynamics  for the asset price $S$ on before-default set $\{t<\tau\}$ for $S^0$ and on after-default set $\{t\geq \tau\}$ for $S^1$:
\begin{eqnarray*}
dS_t^0&=&S_t^0(\mu_t^0dt+\sigma^0_tdW_t), \quad 0 \leq t \leq T  \\
dS_t^1(\theta)&=&S_t^1(\theta)(\mu^1_t(\theta)dt+\sigma^1_t(\theta)dW_t),\quad  \theta \leq t \leq T 
\end{eqnarray*}
where the coefficients {{ $\mu^0$ and  $\sigma^0$ are $\bF$-adapted processes, $\mu^1(\theta)$ and $\sigma^1(\theta)$ are $\bF\otimes\mathcal B(\R_+)$-adapted}} processes, and $W$ is an $\bF$-Brownian motion. In addition,  we suppose the integrability condition
\[\int_0^T\Big|\frac{\mu_t^0}{\sigma_t^0}\Big|^2dt+\int_\theta^T\Big|\frac{\mu_t^1(\theta)}{\sigma_t^1(\theta)}\Big|^2dt+\int_0^T|\sigma^0_t|^2dt+\int_\theta^T|\sigma^1_t(\theta)|^2dt<\infty.\] 
So  the values of the before-default and after-default wealth satisfy  the dynamics 
\begin{eqnarray}
dX_t^{0}(L)&=&X_t^{0}(L)\pi_t^{0}(L)(\mu_t^0dt+\sigma^0_tdW_t),\quad  0 \leq t \leq T    \label{X0}\\
dX_t^{1}(\tau)&=&X_t^{1}(\tau)\pi_t^{1}(\tau)(\mu^1_t(\tau)dt+\sigma^1_t(\tau)dW_t),\quad t \in \lbr \tau, T \rbr \label{X1}
\end{eqnarray}and the jump at default of the wealth process is given by the equality \eqref{X1 initial}.
Finally, we define the admissible trading strategy family $\mathcal A_L$ as the {set of pairs $(\pi^{0}(\cdot),\pi^1(\cdot))$, where $\pi^{0}(\cdot)$ and $\pi^1(\cdot)$ are $\cP(\bF)\otimes\mathcal B(\mathbb R_+)$-measurable processes such that
\begin{equation}\label{set AM}\forall\,{l>0,\;}\bigg(\int_0^{\tau_l\wedge T} |\pi_t^{0}(l)\sigma_t^0|^2dt+\int_{\tau_l\wedge T}^T|\pi^{1}_t(\tau_l)\sigma_t^1(\tau_l)|^2dt\bigg)<\infty \text{ and }  \pi_{\tau_l}^{0}(l)\gamma_{\tau_l}<1, \,a.s.,\end{equation}
where $\tau_l$ is the $\bF$-stopping time defined by $\tau_l:=\inf\{t\,:\,\Lambda_t\geq l\}$.} 


\begin{Rem}\label{rem3.1}
Let $\mathcal A$ denote the set of all $\mathbb G^M$-predictable processes $\pi$ such that
$\int_0^T|\pi_t\sigma_t|^2dt<\infty$ and 
$\pi_{\tau}\gamma_\tau<1.$
If $(\pi^0(\cdot),\pi^1(\cdot))$ is an element in $\mathcal A_L$, then $(\pi_t=\pi^0_t(L)\indic_{\tau\ge t }+\pi_t^1(\tau_L)\indic_{\tau< t},t\geq 0)$
is a processus in the set $\mathcal A$. Conversly, given a process $\pi\in\mathcal A$, there exists a pair $(\pi^0(\cdot),\pi^1(\cdot))\in\mathcal A_L$ such that $\pi_t=\pi_t^0(L)\indic_{\tau\ge t}+\pi^1_t(\tau_L)\indic_{\tau< t}$ for any $t\geq 0$, thanks to Lemma \ref{Lem:decomp}. \end{Rem}


\subsection{The optimization problem}

The insider has the objective to maximize her expected utility function on the terminal wealth of the portfolio. 
Let $U$ be a utility function defined on $(0,+\infty)$, strictly increasing,  strictly concave and of class $C^1$ on $(0,+\infty)$, and satisfying $\lim_{x\rightarrow 0^+}U'(x)= +\infty$ and  $\lim_{x\rightarrow\infty}U'(x)=0$.

We shall consider the problem
\begin{equation}\label{sup U insider}
V_0=\sup_{\pi\in\mathcal A_L} \esp[U(X_T)]
\end{equation}
and search for the optimal strategy $\widehat\pi$ for the insider. A similar problem has been studied in \cite{jp} for a standard investor with $\bG$-predictable strategy $\overline\pi=(\overline\pi^0,\overline\pi^1(\cdot))$ and $\bG$-adapted wealth $\overline{X}=(\overline{X}^0,\overline{X}^1(\cdot))$. The admissible strategy set $\overline{\mathcal A}$ {consists of pairs $\overline\pi=(\overline\pi^0,\overline\pi^1(\cdot))$ where $\overline\pi^0$ and $\overline\pi^1(\cdot)$ are respectively $\mathbb F$-predictable and $\mathcal P(\mathbb F)\otimes\mathcal B(\mathbb R_+)$-measurable, and such that 
\[\forall\, \theta\in\R_+,\; \int_0^T| \overline \pi_t^{0}\sigma_t^0|^2dt+\int_\theta^T| \overline \pi^{1}_t(\theta)\sigma_t^1(\theta)|^2dt<\infty \text{ and }  \overline \pi_\theta^{0}\gamma_\theta<1, a.s.\]}

In this paper, we concentrate on the insider's optimization problem \eqref{sup U insider} and we are interested in the information flow impact on the trading strategies.  Intuitively, the insider should have a larger gain of investment due to the extra  information.  Indeed, 
if the investor and the insider have the same utility function and if they can invest in the same financial assets, then the only difference between them relies on their available information in the sense that the corresponding filtrations satisfy $\bG \subset \bG^M$, which implies the same inclusion for the sets of admissible strategies: $\overline{\mathcal A} \subset  \mathcal A_L$. Thus
 the corresponding value functions satisfy
\begin{equation}\label{comparaisonvaleur}
 \overline V_0:= \sup_{\overline\pi \in\overline{\mathcal A}} \esp[U(\overline{X}_T)] \leq V_0.
\end{equation}
since  the supremum in $ \overline  V_0$ is taken on a smaller set than the one in  $ V_0$. Remark here that on a given sample path, it may happen that $U(\overline{X}_T) >U(X_T)$ if both investor and insider follow their  optimal strategies $\widehat{\overline{\pi}}$ and $\widehat{\pi}$ respectively,  but one always have  in expectation for  optimal strategies $\esp[U(\overline{X}_T^{\widehat{\overline{\pi}}})] \leq \esp[U(X_T^{\widehat{\pi}})]$.

The  {expectation} $\esp[U(X_T)]$ can be written, by using the wealth decomposition  formula \eqref{wealth decomposition}, as
\begin{equation}\label{decomposition UX_T}
\esp[U(X_T)]=\esp[\indic_{T<\tau}U(X_T^{0}(L))+\indic_{T\geq\tau} U(X_T^{1}(\tau))].
\end{equation}
The aim of the above formulation,  similar as in \cite{jp},  is to reduce the initial optimization problem in an incomplete market into two problems : the after-default and before-default ones. Nevertheless, the  approach we adopt here is different since, as mentioned previously, the random time $\tau$ is not a totally inaccessible random time for the insider and we can no longer use the conditional density approach to solve the problem. 

Our approach will use the  theory of initial enlargement of filtration (also called  the strong  information modeling in \cite{Baud2}) by using the value of the random  default barrier $L$   known to the insider.  More precisely, we introduce a family of $\bF$-stopping times 
$\tau_l=\inf\{t: \Lambda_t\geq l\}$ for all $l > 0$ which are  possible realizations of $L$ and we work under an equivalent probability measure $\mathbb{P}^{L}$ under which  $L$ is independent to $\F_T$. Thus in our framework, we shall need the Radon-Nikodym derivative process $p_t(L)$  which is the density of  the historical  probability measure $\mathbb{P}$ with respect to this equivalent  probability measure $\mathbb{P}^{L}$  and it will play a similar role as the default density process in \cite{jp}. 

This probability density hypothesis is given below. It is a standard hypothesis in the theory of initial enlargement of filtration due to Jacod \cite{Jac, Jac2}.
\begin{hypothese}\label{hyp1}
We assume that $L$ is an  $\mathfrak{A}$-measurable random variable
with values in  $]0,+\infty[$, which satisfies the assumption :
$$\mathbb{P}(L\in \cdot\, |\,\mathcal{F}_{t})(\omega) \sim\mathbb{P}(L\in \cdot ), \quad \forall t \in [0,T],\,   \mathbb{P}-a.s..$$
\end{hypothese}
We denote by $P^L _t(\omega,dx)$ a regular version of the
conditional law of $L$ given $\mathcal{F}_t $  and by {$P^{{L}}$}
the law of $L$ (under the probability $\mathbb P$). According to
\cite{Jac2}, there exists a measurable version of the conditional
density
\begin{equation}\label{pt}p_t(x)(\omega)=\frac{dP^L _t}{dP^{L}}(\omega,x)\end{equation}  which is a positive
$(\mathbb F,\mathbb{P})$-martingale
 and hence can be written as
\begin{equation}\label{dynamicspt}
 p_t(x)=p_0(x)+ \int_0^t p_s(x) \rho_s(x) d W_s, \quad\forall x\in]0,+\infty[, \quad t \in [0,T] 
\end{equation}
for some $\mathbb F$-predictable process $(\rho_t(x))_{t \in [0,T]}$.
It is proved in \cite{Pon} that Assumption \ref{hyp1} is satisfied
if and only if there exists a probability measure equivalent to
$\mathbb{P}$ and  under which $\mathcal{F}_T$ and $\sigma(L)$  are independent. Among such equivalent probability measures, the probability
$\mathbb{P}^{L}$ defined by the Radon-Nikodym derivative process
$$\esp_{\mathbb{P}^{L}}  \Big[ \frac{ \mathrm{d}\mathbb{P}}{\mathrm{d}
\mathbb{P}^{L}}\Big|\,\mathcal{G}^M_t\Big]=p_t(L)$$ is the only
one that is identical to   $\mathbb{P}$ on $\mathcal{F}_{\infty}$. For examples of $L$ and explicit computations of corresponding $p_t(L)$, interest reader may refer to \cite{nous2}.

We recall that we consider the optimization problem (\ref{sup U insider})
$V_0=\sup_{\pi\in\mathcal A_L} \esp[U(X_T)]$. {Since the initial $\sigma$-field is non-trivial, it is useful to consider the conditional optimization problem}\begin{equation}\label{sup U insiderbis}
\esssup_{\pi\in\mathcal A_L} \esp[U(X_T)  | \mathcal{G}^M_0 ],
\end{equation}
where $\G_0^M=\sigma(L)$. 
The link between those two optimization problems (\ref{sup U insider}) and (\ref{sup U insiderbis}) is given by \cite{ABS} : if the supremum in (\ref{sup U insiderbis}) is attained by some strategy in $\mathcal A_L$, then the $\omega$-wise optimum is also a solution to (\ref{sup U insider}). Although  the supremum is not {necessarily} attained in our problem, we will see in Proposition \ref{cvversoptimal} that there exists a sequence of admissible strategies $\pi_n$ such that $\esp[U(X^{\pi_n}_T)  | \mathcal{G}^M_0 ]$ converges in $L^1$ to \eqref{sup U insiderbis} and {we can prove} that for the  same sequence, $\esp[U(X^{\pi_n}_T)]$ converges to $V_0$. 
Thus, we will first solve the optimization problem (\ref{sup U insiderbis}), and then deduce the solution of  (\ref{sup U insider}) by taking the expectation, as explained in the following Proposition:

\begin{Pro}\label{Pro:espcond}Under the Assumption \ref{hyp1}, one has 
\begin{align*}
&\esp[U(X_T) | \mathcal{G}^M_0 ]=\esp\left[ {p_T(l)} \left( \indic_{T<\tau_l}U(X_T^{0}(l))+\indic_{T\geq \tau_l} U(X_T^{1}(\tau_l)) \right) \right]_{l=L}\\
&\esp[U(X_T)]=\int_{\R_+}\esp\left[ {p_T(l)} \left( \indic_{T<\tau_l}U(X_T^{0}(l))+\indic_{T\geq \tau_l} U(X_T^{1}(\tau_l)) \right) \right] P^L(dl)
\end{align*}
where for $l > 0$, $\tau_l:=\inf\{t: \Lambda_t\geq l\}$. 
\end{Pro}

\proof 
{We will use the change of probability to $\mathbb{P}^{L}$} in order to reduce to the case where $L$ and $\mathcal F_T$ are independent.  Firstly, by \eqref{decomposition UX_T}, 
\begin{eqnarray*}
\esp[U(X_T)| \mathcal{G}^M_0]&=&  \esp \left[ \indic_{T<\tau}U(X_T^{0}(L))+\indic_{T\geq\tau} U(X_T^{1}(\tau)) \, \, | L \right]  \\
&=& \esp_{\mathbb{P}^{L} } \left[ {p_T(l)} \left( \indic_{T<\tau_l} U(X_T^{0}(l))+\indic_{T\geq\tau_l} U(X_T^{1}(\tau_l)) \right)  \,  \right]_{l=L} \\
&=&  \esp^{ } \left[ {p_T(l)} \left( \indic_{T<\tau_l} U(X_T^{0}(l))+\indic_{T \geq \tau_l} U(X_T^{1}(\tau_l))\right)  \,  \right]_{l=L} 
\end{eqnarray*}
where the last two equalities follow respectively from the facts that $\mathcal{F}_{\infty}$ and $\sigma(L)$  are independent under 
$\mathbb{P}^{L}$ and  that $\mathbb{P}^{L}$  is identical to   $\mathbb{P}$ on $\mathcal{F}_T$. Thus
\begin{eqnarray*}
\esp[U(X_T)] 
&=& \esp  \left[ \esp \left[ \indic_{T<\tau}U(X_T^{0}(L))+\indic_{T\geq\tau} U(X_T^{1}(\tau)) \, \, | L \right]   \right]\\
&=& \int_{\mathbb{R_+}} \esp^{ } \left[ {p_T(l)} \left( \indic_{T<\tau_l} U(X_T^{0}(l))+\indic_{T \geq \tau_l} U(X_T^{1}(\tau_l))\right)  \,  \right] P^L(dl).
\end{eqnarray*}
\finproof

\noindent This motivates to introduce, for any $l > 0$,  the set $\mathcal A_l$  of {pairs $\pi=(\pi^{0},\pi^1(\cdot))$, where $\pi^{0}$ and $\pi^1(\cdot)$ are respectively $\mathbb F$-predictable and $\cP(\bF)\otimes\mathcal B(\mathbb R_+)$-measurable processes, such that
\begin{equation}\label{set Al}\int_0^{\tau_l\wedge T} |\pi_t^{0}\sigma_t^0|^2dt+\int_{\tau_l\wedge T}^T|\pi^{1}_t(\tau_l)\sigma_t^1(\tau_l)|^2dt<\infty \text{ and }  \pi_{\tau_l}^{0}\gamma_{\tau_l}<1, \,a.s.\end{equation} and consider the following optimization problem
\begin{equation}\label{Equ:sup U insiderbis}V_0(l)=\sup_{\pi\in\mathcal A_l}\esp\left[ {p_T(l)} \left( \indic_{T<\tau_l}U(X_T^{0})+\indic_{T\geq \tau_l} U(X_T^{1}(\tau_l)) \right) \right],\end{equation}
where $\tau_l$ is the $\mathbb F$-stopping time $\inf\{t\geq 0\,:\,\Lambda_t\geq l\}$.}

\subsection{The necessity of limiting short selling for the insider before default}
In this subsection, we show that if the regulators do not impose  any constraint on short selling for the insider before the default of the counterparty,
then the insider can achieve a terminal wealth that is not bounded in $L^1$.

\begin{Pro}\label{Thm:uinf} We suppose that the following conditions are satisfied:
\begin{enumerate}[(1)]
\item the process $\Lambda$ is a.s. strictly increasing on $[0,T]$,
\item for any $l$ in the support of the distribution of the law of $L$, one has $\mathbb P(\Lambda_T\geq l)>0$.
\end{enumerate}  
Then we have
\[\esssup_{\pi\in\mathcal A_L} \esp[X_T \,|\, \mathcal{G}^M_0 ]=+\infty\quad\text{a.s.}\]In addition, for any utility function $U$ such that $\displaystyle\lim_{x\rightarrow+\infty}U(x)=+\infty$, 
\[\esssup_{\pi\in\mathcal A_L} \esp[U(X_T) \,|\, \mathcal{G}^M_0 ]=+\infty\quad\text{a.s.}\]

\end{Pro}
\begin{proof}
Let $\varphi:\;]0,+\infty[\rightarrow \;]0,+\infty[$ be an increasing function such that $\varphi(l)<l$ for any $l\in\,]0,+\infty[$. Let $\psi>0$ be a constant. For each $l\in\, ]0,+\infty[$, we define a strategy $\pi(l)=(\pi^0(l),\pi^1(\cdot))\in\mathcal A_l$ as follows
\[\pi^0_t(l)=-\psi\indic_{\tau_{\varphi(l)}<t},\quad
\pi^1(\cdot)\equiv 0.\]
Note that $(\pi^0(\cdot),\pi^1(\cdot))$ is an admissible strategy in $\mathcal A_L$. The value at the time $\tau_l$ of the  corresponding wealth process $X^{\varphi,\psi}(l)$ is equal to $(1 + \gamma_{\tau_l}\psi)X_{\tau_l^-}^{\varphi,\psi}$. By the dynamics of the wealth process \eqref{X0} and \eqref{X1}, on $\{\tau_l\leq T\}$, we have
\[\begin{split}X_{\tau_l}^{\varphi,\psi}&=X_0 (1 + {\gamma}_{\tau_l}\psi) \exp \left(- \int_{\tau_{\varphi(l)}}^{\tau_l}  \big( {\mu^0_t }{\psi} + \frac12 (\sigma^0_t)^2\psi^2\big) dt    -  \int_{\tau_{\varphi(l)}}^{\tau_l}  {\sigma_t^0 }\psi  dW_t   \right).
\end{split} \]
Moreover, 
\[\begin{split}&\quad\;\esp[X_T| \mathcal{G}^M_0]= \esp\left[ {p_T(l)} \left( \indic_{T<\tau_l}X_T^{0}(l)+\indic_{T\geq \tau_l} X_T^{1}(\tau_l) \right) \right]_{l=L}
\geq  \esp\left[\indic_{T\geq\tau_l}p_T(l)X_{\tau_l}^{\varphi,\psi}\right]_{l=L}
\end{split}\]
Now fix an increasing sequence $(\varphi_n)_{n\geq 1}$ of functions such that $\varphi_n(l)<l$ for $l\in ]0,\infty[$ and that $\lim_{n\rightarrow+\infty}\varphi_n(l)=l$. 
By the condition (1), one obtains that $\tau_{\varphi_n(l)}$ converges a.s. to $\tau_l$ when $n\rightarrow+\infty$. 
The sequence of random variables
$\int_{\tau_{\varphi_n(l)}\wedge T}^{\tau_l\wedge T}\sigma_t^0\psi \,dW_t$, $n\geq 1$
converges a.s. to $0$. Then by Fatou's lemma
\[\liminf_{n\rightarrow+\infty}\mathbb E[\indic_{T\geq\tau_l}p_T(l) X_{\tau_l}^{\varphi_n,\psi}]\geqslant \mathbb E[\indic_{T\geq\tau_l}p_T(l)X_0(1+\gamma_{\tau_l}\psi)]\]which implies the first assertion since $\proba(\tau_l\leq T)>0$ by condition (2).
We use the similar argument and  the assumption on $U$ to obtain 
\[\lim_{\psi\rightarrow+\infty}\mathbb E[\indic_{T\geq\tau_l}p_T(l)U(X_0(1+\gamma_{\tau_l}\psi))]_{l=L}=\infty.\]

\end{proof}

\begin{Rem}The strategies mentionned in this proof are not arbitrage strategies because for any fixed function $\varphi$ as in the proof, $\mathbb{P} ( T \in \lbr \tau_{\varphi(L)}, \tau \lbr) >0$
 and on this event, the strategy of the insider that consists of betting on the default before maturity $T$ turns out to be a wrong bet. Thus, on a non null probability set, the strategy of a standard investor outperforms the one of the insiders,  although the converse inequality holds on expectation.
\end{Rem}

Proposition \ref{Thm:uinf} justifies to  consider, instead of $\cA_L$ defined in (\ref{set AM}), the 
following  admissible trading strategy sets for the insider : for all $\delta \geq 0$,
\begin{eqnarray}\label{set AMdelta}
&&\cA^\delta_L=\{(\pi^{0}(\cdot),\pi^{1}(\cdot))\in\mathcal A_L  \text{ such that  } \pi^{0}\geq - \delta  \}  \end{eqnarray}
and to quantify the impact of the lower bound $\delta$. {Similar to \eqref{set Al}}, we define
\[\mathcal A_l^\delta=\{(\pi^0,\pi^1(\cdot))\in
\mathcal A_l\text{ such that }\pi^0\geq -\delta\}.\]

{Before considering  the main optimization problem 
\begin{equation}\label{Equ:optdel}
\esssup_{\pi\in\mathcal A^\delta_L}\esp[U(X_T)\, |\, \mathcal{G}^M_0 ]
\end{equation}
we first consider an alternative family of optimization problems depending on the parameter $l\in\mathbb ]0,+\infty[$, 
\begin{equation}\label{optimisation with tau_l} 
V_0^\delta(l):=\sup_{\pi\in\cA_l^\delta} \esp\left[ {p_T(l)} \left( \indic_{T<\tau_l}U(X_T^{0}(l))+\indic_{T\geq \tau_l} U(X_T^{1}(\tau_l)) \right) \right]. 
\end{equation}
The following theorem shows that the optimal value of the optimization problem \eqref{Equ:optdel} is actually equal to $V_0^\delta(L)$.}


\begin{Thm}\label{Thm:measel} 
With the above notation, we have
\[V_0^\delta(L)=\esssup_{\pi\in\mathcal A^\delta_L}\esp[U(X_T)\, |\, \mathcal{G}^M_0 ]
\qquad\text{a.s.}\]
\end{Thm}
\begin{proof}
Assume that $(\pi^0(\cdot),\pi^1(\cdot))$ is an element in $\mathcal A_L^\delta$, then $(\pi^0(l),\pi^1(\cdot))\in\mathcal A_l^\delta$. By Proposition \ref{Pro:espcond} we obtain that 
\[\esssup_{\pi\in\mathcal A^\delta_L}\esp[U(X_T)\, |\, \mathcal{G}^M_0 ]\leq  V_0^\delta(L).\]
For the converse inequality, we shall use measurable selection theorem. For any $\varepsilon>0$ and any $l\in\mathbb ]0,\infty[$, let $F_\varepsilon(l)$ be the set of  strategies $(\pi^0,\pi^1(\cdot))\in\mathcal A_l^\delta$ which are $\varepsilon$-optimal with respect to the problem \eqref{optimisation with tau_l}, namely such that 
\[\esp\left[ {p_T(l)} \left( \indic_{T<\tau_l}U(X_T^{0}(l))+\indic_{T\geq \tau_l} U(X_T^{1}(\tau_l)) \right) \right]\geq\begin{cases}
V_0^\delta(l)-\varepsilon,&\text{if }V_0^\delta(l)<+\infty,\\
1/\varepsilon,&\text{if }V_0^\delta(l)=+\infty.
\end{cases}\]
By a measurable selection theorem (cf. Benes \cite[ Lemma 1]{Benes}), there exists a measurable (with respect to $l$) family $\{(\pi^0(l),\pi^1(\cdot,l))\}_{l\in\mathbb R_+}$ such that $(\pi^0(l),\pi^1(\cdot,l))\in F_\varepsilon(l)$ for any $l >0$. Finally let 
\[\widetilde\pi^0(\cdot):=\pi^0(\cdot),\quad\widetilde\pi^1_t(x):=\indic_{t> x}\pi^1_t(x,\Lambda_x).\]
We have $(\widetilde\pi^0(\cdot),\widetilde{\pi}^1(\cdot))\in\mathcal A_L$ and $\widetilde\pi^1_t(\tau_l)=\pi^1_t(\tau_l,l)$ for any $l >0$ on $\{\tau_l< t\}$. Therefore, by Proposition \ref{Pro:espcond}, 
\[\begin{split}\esp [U(\widetilde X_T)\,|\,\mathcal G_0^M]&=\esp\left[ {p_T(l)} \left( \indic_{T<\tau_l}U(X_T^{0}(l))+\indic_{T\geq \tau_l} U(X_T^{1}(\tau_l)) \right) \right]_{l=L}\\
&\geq\begin{cases}
V_0^\delta(L)-\varepsilon,&\text{if }V_0^\delta(L)<+\infty,\\
1/\varepsilon,&\text{if }V_0^\delta(L)=+\infty.
\end{cases}
\end{split}\]
Since $\varepsilon$ is arbitrary, we obtain $\esp[U(\widetilde X_T)\,|\,\mathcal G_0^M]\geq V_0^\delta(L)$.\finproof
\end{proof}

%
%
%

\section{Solving the optimization problems}

%

In this section, we concentrate on solving the optimization problem \eqref{optimisation with tau_l}  
$$ \sup_{\pi \in\cA^\delta_l} \esp\left[ {p_T(l)} \left( \indic_{T<\tau_l}U(X_T^{0}(l))+\indic_{T\geq \tau_l} U(X_T^{1}(\tau_l)) \right)  \right]$$
for any fixed $l >0$.  We recall that the before-default  and after-default wealth processes $X^0$ and $X^1$ are governed by two control processes $\pi^0$ and $\pi^1$ respectively, so we need to search for a couple of optimal controls $\hat\pi= (\hat\pi^0,\hat\pi^1)$. In Theorem \ref{Thdecomposition} we explain how to decompose the optimization problem into two problems each depending only on $\pi^0$ and on $\pi^1$ respectively.

 \noindent The  after-default optimization problem will be solved firstly, using the after-default filtration $\mathbb{F}^1$ 
 $$\mathbb{F}^1:= (\F_{{\tau_l} \vee t} )_{t \in [ 0,T]}.$$
 Remark that the initial $\sigma$-field of the filtration $\mathbb{F}^1$ is not trivial :  $$\F^1_0=\F_{\tau_l}=\{A\in\mathfrak A: A\cap\{\tau_l\leq t\}\in\F_t,t\in [ 0,T]\}$$ and $\tau_l$ is $\F_{\tau_l}$-measurable. 
All the $\mathbb{F}^1$-adapted processes involved in the after-default optimization problem are indexed on the {right-upper side by the symbol ``$1$''}. In particular,  we  denote by $X^{1,x_l}(\tau_l)$ the solution of the SDE \eqref{X1} defined on the stochastic interval $\lbr
\tau_l,T\rbr$ starting from the $\bF$-stopping time $\tau_l$ {with $\mathcal{F}_{\tau_l}$-measurable initial value $x_l$}. We define by $\cA^1_{l}$   the admissible predictable strategy set  $(\pi^{1}_t(\tau_l), t\in\,\rbr
\tau_l,T\rbr)$ such that $\int_{\tau_l }^{\tau_l\vee T}|\pi_t^{1}(\tau_l)\sigma_t^1(\tau_l)|^2dt<\infty$ a.s.. \\

\noindent  The global before default optimization problem involves the solution of the after-default optimization problem and will be solved in a second step, using the stopped filtration $\mathbb{F}^0$ 
 $$\mathbb{F}^0:= (\F_{{\tau_l} \wedge t} )_{t\in [ 0,T]}.$$
All the $\mathbb{F}^0$-adapted processes involved in the after-default optimization problem are indexed on the {right-upper side by the symbol ``$0$''}. The admissible predictable strategy set $\cA^{0,\delta}_l$ is $(\pi^{0}_t(l),  t\in\lbr0 ,\tau_l\wedge T\rbr  )$ such that $\int_0^{\tau_l \wedge T}|\pi_t^{0}(l)\sigma_t^0|^2dt<\infty$, $\pi_t^{0}(l)\geq -\delta$ and $1>\pi^{0}_{\tau_l}(l)\gamma_{\tau_l}$, a.s.. \\

\begin{Thm}\label{Thdecomposition}
For any $l>0$, we denote by $V_{\tau_l}^{1}(x_l)$ the optimal value of the after-default optimization problem 
\begin{equation}\label{V after default}
V_{\tau_l}^{1}(x_l)=\esssup_{\pi^{1}(\tau_l)\in\cA^1_{l}} \esp[p_T(l) U(X_T^{1,x_l}({\tau_l}))|\F_{\tau_l}],
\end{equation}
then the global optimal value $V_0^\delta(l)$ defined in \eqref{optimisation with tau_l}  can be written as the solution of a
before-default optimization problem as
\begin{equation}\label{V global}
V_0^\delta(l)=\sup_{\pi^{0}\in\cA^{0,\delta}_l}\esp \left[   \indic_{T<\tau_l} p_T(l) U(X_T^{0}(l))+\indic_{T\geq \tau_l} V^{1}_{\tau_l}\big(X_{\tau_l}^{0}(l)(1-\pi^{0}_{\tau_l}(l)\gamma_{\tau_l})\big)  \right].
\end{equation} 
\end{Thm}

\begin{proof}  {Consider firstly an arbitrary admissible strategy $(\pi^0,\pi^1(\cdot))\in\mathcal A_l^\delta$ for a fixed $l > 0$. By definition, $(\pi^0_t,\,t\in\lbr 0,\tau_l\wedge T\rbr)\in\mathcal A_l^{0,\delta}$ and $(\pi^1_t(\tau_l),\,t\in\,\rbr\tau_l,T\rbr)\in\mathcal A_{l}^1$.}
{Taking the conditional expectation with respect to $\F_{\tau_l}$ leads to the following inequalities
\[\begin{split}
&\quad\;\esp \left[ {p_T(l)} \left( \indic_{T<\tau_l}U(X_T^{0}(l))+\indic_{T\geq \tau_l} U(X_T^{1}(\tau_l)) \right) \right] \\
&= \esp\left[   \indic_{T<\tau_l} p_T(l) U(X_T^{0}(l))+  \indic_{T\geq \tau_l}  \esp  \left(  p_T(l) U(X_T^{1}(\tau_l)) | \F_{\tau_l}   \right) \right] \\
&\leq  \esp\left[   \indic_{T<\tau_l} p_T(l) U(X_T^{0}(l))+  \indic_{T\geq \tau_l}  V_{\tau_l}^{1}(X_{\tau_l}^{0}(l)(1-\pi^{0}_{\tau_l}(l)\gamma_{\tau_l})) \right] \\
&\leq V_0^\delta(l). 
\end{split}\]
}

	For the converse inequality, let us assume for the moment that the  esssup  in the definition (\ref{V after default}) is achieved for a given $\hat \pi^{1}(\tau_l) $ (see section \ref{afterdefault} for the proof). Then for any   $(\pi^{0}_t,  t \in\lbr 0, \tau_l\rbr  )$ in $ \cA^{0,\delta}_l$, by a measurable selection theorem, there exists a $\mathcal P(\mathbb F)\otimes\mathcal B(\mathbb R_+)$-measurable process $\pi^1(\cdot)$ such that $\pi^1(\tau_l)=\hat\pi^1(\tau_l)$ on $\rbr \tau_l,T\rbr$ and that $(\pi^0,\pi_1(\cdot))\in {\mathcal A_l^\delta}$, where we have extended $\pi^0$ to an $\mathbb F$-predictable process {on $\mathbb R_+$}. Thus
\begin{eqnarray*}
 V_0^\delta(l) & \geq &  \esp\left[   \indic_{T<\tau_l} p_T(l) U(X_T^{0}(l))+  \indic_{T\geq \tau_l}  V_{\tau_l}^{1}(X_{\tau_l}^{0}(l)(1-\pi^{0}_{\tau_l}(l)\gamma_{\tau_l})) \right] 
\end{eqnarray*}
By taking the supremum over all $(\pi^{0}_t(l), t \in\lbr 0, \tau_l\rbr  )$ $ \in  \cA^{0,\delta}_l$,  we obtain the desired inequality.\finproof
\end{proof}

\begin{Rem}
The supremum  in $V_0^\delta(l)$ can be approached by a sequence of admissible strategies in $\cA^{0,\delta}_l$ (see Proposition \ref{cvversoptimal}), which induces a sequence of strategies in $\cA^{\delta}_L$
such that the corresponding  value functions  converge to  $V_0^{\delta}:=\sup_{\pi\in\mathcal A^{\delta}_L} \esp[U(X_T)]$.  Thus  $V_0^{\delta}=\int_{\R_+} V_0^\delta(l) P^L(dl)$.
\end{Rem}

\begin{Rem}\label{Rklebesgue}
The process $(p_t(l),t \in [ 0,T])$ is essential in our approach of initial information and plays a similar role to the default density process in \cite{jp} $(\alpha_t(\theta),t \in [ 0,T])$ defined as $\alpha_t(\theta)d\theta=\proba(\tau\in d\theta|\F_t)$.  Thus,  to quantify the gain of the insider (as in the forthcoming  Proposition  \ref{propcomparisonafter}),   it is interesting to compare those two processes. 
\begin{itemize}
\item In the particular case where the $\F_t$-conditional law of $L$ admits a density $g_t(l)$ w.r.t. the Lebesgue measure,  the default density can be completely deduced (see \cite[Proposition 3]{ejjz}) in this framework as $\alpha_t(\theta)=\lambda_\theta g_t(\Lambda_\theta)$, $t\geq\theta$ and $\alpha_t(\theta)=\esp[\alpha_{\theta}(\theta)|\F_t]$, $t\leq\theta$ where $\lambda$  is the process given in Section 2. 
\item In the general case, the law of $L$ can have atoms, then the default density does not exist and the approach in \cite{jp} is no longer valid, whereas  the insider's optimization problem can be solved with the process $p_{.}(l)$.
 \end{itemize}\end{Rem}

\subsection{The after-default optimization}\label{afterdefault}
In this section, we focus on the after default optimization problem \eqref{V after default} in the  filtration $\mathbb{F}^1= (\F_{{\tau_l} \vee t} )_{t \in [ 0,T]}$
$$V_{\tau_l}^{1}(x_l)=\esssup_{\pi^{1}(\tau_l)\in\cA^1_{l}} \esp[p_T(l) U(X_T^{1,x_l}({\tau_l}))|\F_{\tau_l}] $$
where $\tau_l$  is an $\bF$-stopping time and the initial wealth  $x_l$ is $\mathcal{F}_{{\tau_l}}$-measurable.
This problem  is similar to  a standard optimization problem, we will extend the results in our framework where the initial time $\tau_l$ is a random time  (and is an $\bF$-stopping time).  
We define the process
$$Z_t(\tau_l)=\exp\Big(-\int^{{\tau_l} \vee t}_{\tau_l}\frac{\mu_u^1(\tau_l)}{\sigma_u^1(\tau_l)}dW_u-\frac 12 \int^{{\tau_l} \vee t}_{\tau_l}\left|\frac{\mu_u^1(\tau_l)}{\sigma_u^1(\tau_l)}\right|^2du\Big), \, \,   \, \, t \in [ 0,T] .$$
This process is an $\mathbb{F}^1$-local martingale (cf. \cite[page 20]{KaS}), we assume that the coefficients  $\mu^1(\tau_l)$ and $\sigma^1(\tau_l)$ satisfy a Novikov criterion (see Theorem \ref{them 4.4} below) so that $(Z_t(\tau_l))_{t \in [ 0,T]}$ is an $\mathbb{F}^1$-martingale.
\begin{Thm}\label{them 4.4}
We assume that for any $l\geq 0$, the coefficients  $\mu_u^1(\tau_l)$ and $\sigma_u^1(\tau_l)$ satisfy the  Novikov criterion \[\esp\bigg[\exp\bigg(\frac{1}{2} \int_{\tau_l}^{\tau_l\vee T} \left| \frac{\mu_u^1(\tau_l)}{\sigma_u^1(\tau_l)}\right|^2du \bigg)\bigg]< \infty.\]
Then the value function process to problem \eqref{V after default} is a.s. finite and is given by
\begin{equation} \label{ValueFunction1initial}
\hat V^{1}_{\tau_l}(x_l)=\esp\bigg[ p_T(l) U\bigg(I\big(\hat y_{\tau_l}(x_l)\frac{ Z_T(\tau_l)}{p_T(l)}\big)\bigg)\,\bigg|\,\F^1_{0}\bigg]
\end{equation}
where $I=(U')^{-1}$ and the Lagrange multiplier $\hat y_{\tau_l}(\cdot)$ is the unique $\F_{\tau_l} \otimes\mathcal B(\R_+)$-measurable solution of the equation 
\[\frac{1}{Z_t(\tau_l)}\esp \left[Z_T(\tau_l)     I\big(\hat y_{\tau_l}(x_l)\frac{ Z_T(\tau_l)}{p_T(l)}\big)      \,\bigg|\,\F^1_0 \right]=x_l.\] 
The corresponding optimal wealth is equal to 
\begin{equation}\label{wealth}\hat X_t^{1,x_l}(\tau_l)=\frac{1}{Z_t(\tau_l)}\esp \left[Z_T(\tau_l)     I\big(\hat y_{\tau_l}(x_l)\frac{ Z_T(\tau_l)}{p_T(l)}\big)      \,\bigg|\,\F^1_t \right], \quad t\in\lbr\tau_l, T\rbr.
\end{equation}

\end{Thm}

\proof  
Note that after the default, the market is complete. The process $ Z(\tau_l) X^{1,x_l}({\tau_l})$ is a  positive local $\mathbb{F}^1$-martingale, and thus a supermartingale, leading to the following budget constraint :
\begin{equation}\label{CB}
\esp \left(   Z_T(\tau_l) X_T^{1,x_l}({\tau_l}) \,\Big|\,\mathcal{F}^1_{0} \right)  \leq x_l. 
\end{equation}
Conversely, the  martingale representation theorem on the brownian filtration implies that for any $\mathcal{F}_T \otimes \mathcal{B}(\mathbb{R}_+)$-measurable  $X_T(\cdot)$, there exists a  $\mathcal{P}(\bF) \otimes \mathcal{B}(\mathbb{R}_+)$- measurable process $\phi(\cdot)$ such that 
$$ X_T(\tau_l) \indic_{\tau_l <T}  = \esp ( X_T(\tau_l)  \indic_{\tau_l <T} |\mathcal{F}_{{\tau_l}} ) + \int^{\tau_l \vee T}_{\tau_l} \phi_u(\tau_l) du
$$
Therefore the after default optimization problem is solved by the mean of the Lagrange multiplier 
$$
V^{1}_{\tau_l}(x_l)=\esp\bigg[ p_T(l) U\bigg(I\big(\hat y_{\tau_l}(x_l)\frac{ Z_T(\tau_l)}{p_T(l)}\big)\bigg)|\F^1_{0}\bigg]$$
and the optimal wealth is given by
\[\hat X_t^{1,x_l}(\tau_l)=\frac{1}{Z_t(\tau_l)}\esp \left[Z_T(\tau_l)     I\big(\hat y_{\tau_l}(x_l)\frac{ Z_T(\tau_l)}{p_T(l)}\big)      |\F^1_{ t}  \right]\]
where $I=(U')^{-1}$ and {{ the Lagrange multiplier $\hat y_{\tau_l}(x_l)$ is $\F_{\tau_l} \otimes\mathcal B(\R_+)$-measurable and satisfies}}  $\hat X_{\tau_l}^{1,x_l}(\tau_l)=x_l$. The {existence, uniqueness and measurability} of the Lagrange multiplier $\hat y_{\tau_l}(x_l)$ in the case of a non trivial initial $\sigma$-field is proved in Proposition 4.5 of Hillairet  \cite{Hil}.
\finproof

We can already give a first analytical comparison  of  the value function process of the after default optimization problem between the initial information and the progressive information. We recall (see \cite[Theorem 4.1]{jp}) that for the progressive information, this value function process at the default time $\theta=\tau_l$ is
\begin{eqnarray}\label{ValueFunction1progressive}
\bar V^{1}_{\tau_l}(x)&=&\esp\bigg[ \alpha_T(\tau_l) U\bigg(I\big(\bar y_{\tau_l}(x)\frac{ Z_T(\tau_l)}{\alpha_T(\tau_l)}\big)\bigg)\,\bigg|\,\F_{\tau_l}\bigg]\\
&& {\mbox{ with }}  \overline y_{\tau_l} {\mbox{ such that }} 
\esp \left[ \frac{Z_T(\tau_l) }{  Z_t(\tau_l) }    I\big(\bar y_{\tau_l}(x)\frac{ Z_T(\tau_l)}{\alpha_T(\tau_l)}\big)      \,\bigg|\,\F_{ \tau_l}  \right] =x \nonumber
\end{eqnarray}
which can be compared to the value function (\ref{ValueFunction1initial}) for an initial information.
To do this, we assume that the $\mathcal{F}_t$-conditional law of $L$ admits a density with respect to the Lebesgue measure, denoted as $g_t(l)$ (see Remark \ref{Rklebesgue}). Using Assumption \ref{hyp1}, it is sufficient to assume   that the law of $L$ admits a density with respect to the Lebesgue measure.

\begin{Pro}\label{propcomparisonafter}
The value function processes of the after-default optimization problem  for the initial information $\hat V^{1}_{\tau_l}(x)$ and  for   the progressive information  $\bar V^{1}_{\tau_l}(x)$ satisfy
$$\frac{\bar V^{1}_{\tau_l}(x)}{ \hat V^{1}_{\tau_l}(x)}= g_0(l) \lambda_{\tau_l} $$
where $g_0(.)$ denotes the density of the law of $L$ with respect to the Lebesgue measure and is supposed to be strictly positive.
In addition, the ratio $\frac{\alpha_T(\tau_l)}{p_T(l)}$ is $\F_{\tau_l}$-measurable and is also equal to $\frac{\bar V^{1}_{\tau_l}(x)}{ \hat V^{1}_{\tau_l}(x)}$.
\end{Pro}

\proof
On the stochastic interval $\lbr \tau_l,T\rbr$, the density of probability $Z_t(\tau_l)$ is the same for the initial and the progressive information, but not necessarily the Lagrange multipliers, although they satisfy the  same type of equation. More precisely,  $\bar y_{\tau_l}(x)$ is the unique  $\F_{\tau_l} \otimes\mathcal B(\R_+)$-measurable solution of 
\begin{equation}\label{aa}\esp \left[ \frac{Z_T(\tau_l) }{  Z_t(\tau_l) }    I\big(\bar y_{\tau_l}(x)\frac{ Z_T(\tau_l)}{\alpha_T(\tau_l)}\big)      |\F_{ \tau_l}  \right] =x\end{equation} and 
$\hat y_{\tau_l}(x)$ is the unique  $\F_{\tau_l} \otimes\mathcal B(\R_+)$-measurable solution of 
\begin{equation}\label{bb}\esp \left[ \frac{Z_T(\tau_l) }{  Z_t(\tau_l) }    I\big(\hat y_{\tau_l}(x)\frac{ Z_T(\tau_l)}{p_T(l)}\big)      |\F_{ \tau_l}  \right] =x.\end{equation} 
We recall that $g_T(l)$ is the density of the $\mathcal{F}_T$-conditional law of $L$  with respect to the Lebesgue measure. According to \eqref{pt} and Remark \ref{Rklebesgue},  $p_T(l)=\frac{g_T(l)}{g_0(l)}$ and $\alpha_T(\tau_l)=\lambda_{\tau_l} g_T(\Lambda_{\tau_l})=\lambda_{\tau_l} g_T(l) $. Thus the ratio $\frac{\alpha_T(\tau_l)  }{ p_T(l)}= \lambda_{\tau_l}g_0(l)$  is $\F_{ \tau_l}$-measurable. Furthermore, \eqref{aa} is equivalent to \[\esp \left[ \frac{Z_T(\tau_l) }{  Z_t(\tau_l) }    I\big(\bar y_{\tau_l}(x)\frac{p_T(l)}{\alpha_T(\tau_l)}\frac{ Z_T(\tau_l)}{p_T(l)}\big)      |\F_{ \tau_l}  \right] =x\] so the unicity of the $\F_{\tau_l}$-measurable solution of the equation \eqref{bb} implies that the  $\bar y_{\tau_l}(x)\frac{p_T(l)}{\alpha_T(\tau)}=\hat y_{\tau}(x)$. Therefore,   $\frac{ \bar y_{\tau_l}(x)  }{ \hat y_{\tau_l}(x)}= \lambda_{\tau_l}g_0(l)$ and 
$U\big(I\big(\bar y_{\tau_l}(x)\frac{ Z_T(\tau_l)}{\alpha_T(\tau_l)}\big)\big)=U\big(I\big(\hat y_{\tau_l}(x)\frac{ Z_T(\tau_l)}{p_T(l)}\big)\big)$. We conclude the proof by using again that the ratio $\frac{\alpha_T(\tau_l)  }{ p_T(l)}= \lambda_{\tau_l}g_0(l)$  is $\F_{ \tau_l}$-measurable.
\finproof

\begin{Rem}\label{remarque atom}
\begin{enumerate}
 \item 
The constant $\frac{1}{g_0(l)}$ is a threshold to compare with $\lambda_{\tau_l}$. For a given scenario, if $\lambda_{\tau_l}$ is smaller than $\frac{1}{g_0(l)}$, which is often the case in practice, then the absolute value of the  value function for the initial information is greater than the one for the  progressive information. 
In this case, the insider puts a higher weight $p_T(l)$ on the after-default optimization problem \eqref{V after default}, compared to a standard investor whose weight is $\alpha_T(\tau_l)$.  This means that $p_t(l)$ conveys more information than $\alpha_t(\tau_l) $. Moreover, the smaller $g_0(l)$ is, the  greater  is the gain of an insider. The interpretation is as follows:  small value of $g_0(l)$ means that $\mathbb{P}(L \in dl)$  is small, thus the  insider, who knows the real value of $L$, has very relevant information.   

\item Concerning the optimal wealth $\hat X^1(\tau_l)$ in \eqref{wealth}, we can prove, using the same argument as in the proof of Proposition \ref{propcomparisonafter} that  starting from a same wealth $x_l$ at the default time $\tau_l$, the optimal  wealth process of  the after-default optimization problem is the same for the initial and the progressive information, i.e. $\hat X^{1,x_l}(\tau_l)=\bar X^{1,x_l}(\tau_l)$. This result is not surprising, since after the  default, the two information flows coincide.  But naturally the input wealth of the   after-default optimization problem will not be the same for the two information flows since they are not the same before $\tau_l$.
\end{enumerate}
\end{Rem}

\noindent 
We will quantify numerically the gain of an insider compared to a standard investor for the global optimization problem. We now consider, as in \cite{jp}, Constant Relative Risk Aversion (CRRA)  utility functions 
$$
U(x) = \frac{x^p}{p}, \;\;\;\;\;  0<p<1,   \; x >0
$$
and $I(x)=x^{\frac{1}{p-1}}$. Direct computations from the previous theorem yield the optimal wealth

$$\hat X_t^{1,x_l}(\tau_l)=\frac{x_l}{  Z_t(\tau_l)  }
 \frac{\esp \left[p_T(l) \left( \frac{ Z_T(\tau_l)}{p_T(l)}\right) ^{\frac{p}{p-1}} \,\Big|\,\F^1_{t} \right] }{\esp \left[p_T(l) \left( \frac{ Z_T(\tau_l)}{p_T(l)}\right) ^{\frac{p}{p-1}} \,\Big|\,\F_{\tau_l} \right] },\quad t\in\lbr\tau_l, T\rbr        $$ 

and the optimal value function
\begin{equation}\label{V1 CRRA}
\hat V^{1}_{\tau_l}(x_l)= \frac{x_l^p}{ p } \left(    \esp \left[p_T(l) \left( \frac{ Z_T(\tau_l)}{p_T(l)}\right) ^{\frac{p}{p-1}} \,\bigg|\,\F_{\tau_l} \right]            \right)^{1-p} 
=:  \frac{x_l^p}{ p } K_{\tau_l}
\end{equation}
where the $\F_{\tau_l}$-measurable random variable $K_{\tau_l}=\left(    \esp \left[p_T(l) \left( \frac{ Z_T(\tau_l)}{p_T(l)}\right) ^{\frac{p}{p-1}} \,\bigg|\,\F_{\tau_l} \right]            \right)^{1-p}$ only depends on the stopping time $\tau_l$ and on market parameters.

\subsection{The before-default optimization}
We now consider the optimization problem \eqref{V global} with CRRA utility functions.  Using  \eqref{V1 CRRA}, we have to solve :

\begin{equation}\label{globalCRRA}
 {V^{\delta}_0}(l)=\sup_{\pi^{0}\in\cA_l^{0,\delta}}\esp\bigg[\indic_{T<\tau_l} p_T(l) U( X_T^{0}(l)) + \indic_{T\geq \tau_l}   K_{\tau_l} U\big( X^{0}_{\tau_l}(l)(1-\pi^{0}_{\tau_l}(l)\gamma_{\tau_l}) \big) \bigg]
\end{equation}
where the $\F_{\tau_l}${-measurable} random variable $K_{\tau_l}$ does not depend on the control process $\pi^{0}\in\cA_l^{0,\delta}$.

We will use a dynamic programming approach. Recall that  $\mathbb{F}^0=(\F_{{\tau_l} \wedge t } )_{t \in [ 0,T]}$ is  the stopped filtration at the default time. 
Since $\indic_{t < \tau_l}$ is $ \F_{{\tau_l} \wedge t}$-measurable, we have
\begin{align*}
&\quad \esp\big[\indic_{T<\tau_l} p_T(l) U( X_T^0(l)) + \indic_{T\geq \tau_l}   K_{\tau_l} U\big( X^{0}_{\tau_l}(l)(1-\pi^{0}_{\tau_l}(l)\gamma_{\tau_l}) \big)      |  \F_{{\tau_l} \wedge t} \big]\\
&=\indic_{t \geq \tau_l} K_{\tau_l} U\left( X^{0}_{\tau_l}(l)(1-\pi^{0}_{\tau_l}(l)\gamma_{\tau_l})\right) \\
& \, \, \, \, \, \, + \esp\big[\indic_{T<\tau_l} p_T(l) U( X_T^0(l))
 + \indic_{ t< \tau_l \leq T}   K_{\tau_l} U\big( X^{0}_{\tau_l}(l)(1-\pi^{0}_{\tau_l}(l)\gamma_{\tau_l})\big )       |  \F_{{\tau_l} \wedge t} \big]
\end{align*}
For any $\nu$ $\in$ $\cA_l^{0,\delta}$, we introduce   the family of
$\mathbb{F}^0$-adapted processes
\begin{equation}
\mathcal X_t(\nu):=\esssup_{\pi^{0}\in \cA_l^{0,\delta}(t,\nu)}
\esp\bigg[\indic_{T<\tau_l} p_T(l) U( X_T^0(l)) + \indic_{ t< \tau_l \leq T}   K_{\tau_l} U\big( X^{0}_{\tau_l}(l)(1-\pi^{0}_{\tau_l}(l)\gamma_{\tau_l})  \big)     |  \F_{{\tau_l} \wedge t} \bigg]
\end{equation}
where $ \cA_l^{0,\delta}(t,\nu)$ is  the set of controls coinciding with $\nu$
until time $t$: for any $t$ $\in$ $[0,T]$, $\nu$ $\in$ $\cA_l^{0,\delta}$,
$\cA_l^{0,\delta}(t,\nu)=\{ \pi^{0} \in \cA_l^{0,\delta}: \pi^{0}_{.\wedge t} = \nu_{.\wedge t} \}$,
$X^{\nu,0}$ denotes the wealth process derived from the control $\nu$ $\in$ $\cA_l^{0,\delta}$.
We have $V_0^{\delta}(l)=\mathcal X_0(\nu) $ for any $\nu$ $\in$ $\cA_l^{0,\delta}$. 

In the following result, we show that the short selling constraint plays a crucial role in the optimization. In fact, it is the optimal strategy at the default time.  
\begin{Pro}\label{cvstrategie}
For any $\pi^{0}\in\cA_l^{0,\delta}$, there exists a sequence of strategies $ (\pi^0_n \in \cA_l^{0,\delta})_{n\in\mathbb N^*}$ such that 
$\pi^{0}_{n, \tau_l} = -\delta$  and   $$\lim_{n \rightarrow + \infty}  \mathcal X_0(\pi^{0}_n ) \geq \mathcal X_0(\pi^{0}  ).$$
\end{Pro}

\proof
Let $(\tau_n)_{n \in \mathbb{N}^*}$ where $\tau_n<\tau_l$ be an increasing sequence of $\mathbb{F}$-stopping times  that converge to $\tau_l$.
Starting from a strategy $ \pi^0  \in \cA^{0,\delta}_l$, we define another strategy $ \pi^0_n =\indic_{\lbr 0, \tau_n \rbr} \pi^0 - \indic_{\rbr \tau_n, \tau_l \rbr } \delta $
 that remains in $\cA^{0,\delta}_l$. We denote as $X^0(l)$ and $X_n^0(l)$ the corresponding wealth before-default, and as 
$\mathcal X_0(\pi^{0})$ and $\mathcal X_0(\pi_n^{0})$ the corresponding value function for those strategies of the before default global optimization problem.
On the one hand,  by dominated convergence theorem, it is easy to check that
$$\lim_{n \rightarrow + \infty} \esp\big[\indic_{T<\tau_l} p_T(l) |U( X_T^{0}(l))  - U( X_{n,T}^{0}(l))|   \big] =0.$$
On the other hand, on the event 
$\{T\geq \tau_l\}$
{\small 
 $$ \frac{X_{n,\tau_l}^{0}(l) (1-\pi^{0}_{n,\tau_l}(l)\gamma_{\tau_l}) }{ X_{\tau_l}^{0}(l)(1-\pi^{0}_{\tau_l}(l)\gamma_{\tau_l}) }=
\frac{ (1+\delta\gamma_{\tau_l}) }{(1-\pi^{0}_{\tau_l}(l)\gamma_{\tau_l}) }
\exp\left( \int_{\tau_n}^{\tau_l}   (-(\pi_s^0+\delta) \mu^0_s     +\frac{1}{2} (\sigma_s^0)^2(   (\pi_s^0)^2  -\delta^2    )  ) ds -   \int_{\tau_n}^{\tau_l} (\pi_s^0+\delta) \sigma^0_s   dW_s       \right)$$}
 $ \pi^0  \in \cA^{0,\delta}_l$ implies that  $\pi^{0}_{\tau_l}(l) \geq -\delta$ and  
$\frac{ (1+\delta\gamma_{\tau_l}) }{(1-\pi^{0}_{\tau_l}(l)\gamma_{\tau_l}) }\geq 
 1  $ (because $\gamma>0$).
$\tau_n  \rightarrow  \tau_l$ implies that  the exponential term tends to $1$ a.s. 
Thus $ \lim_{ n \rightarrow + \infty } \frac{X_{n,\tau_l}^{0}(l) (1-\pi^{0}_{n,\tau_l}(l)\gamma_{\tau_l}) }{ X_{\tau_l}^{0}(l)(1-\pi^{0}_{\tau_l}(l)\gamma_{\tau_l}) }  \geq 
 1   $ and 
$$\lim_{n \rightarrow + \infty}
\esp\bigg[ \indic_{T\geq \tau_l}   K_{\tau_l}   U( X^{0}_{n,\tau_l}(l)(1-\pi^{0}_{n, \tau_l}(l)\gamma_{\tau_l}) )  \bigg] \geq 
 \esp\bigg[\indic_{T\geq \tau_l}   K_{\tau_l} U( X^{0}_{\tau_l}(l)(1-\pi^{0}_{\tau_l}(l)\gamma_{\tau_l}) ) \bigg].$$
Consequently, 
    $$\lim_{n \rightarrow + \infty}  \mathcal X_0(\pi^{0}_n ) \geq \mathcal X_0(\pi^{0}  ).$$

\finproof

We now characterize the optimal strategy process. From the dynamic programming principle, the following result holds:

\begin{Lem}\label{dynprosurmar}For any $\nu\in\cA_l^{0,\delta}$, the process
$$\xi_t^\nu:=\mathcal X_t(\nu)+\indic_{t \geq \tau_l} K_{\tau_l} U( X^{\nu,0}_{\tau_l}(l)(1-\nu_{\tau_l}(l)\gamma_{\tau_l})), \, \, \,  \, \, \, 0 \leq  t \leq T$$
is an $\mathbb{F}^0$-supermartingale. Furthermore, the optimal strategy  $\widehat{\pi}^{0}$ is characterized by the martingale  property : $(\xi_t^{\widehat{\pi}^{0}})_{0 \leq  t \leq T}$ is 
an $\mathbb{F}^0$-martingale.
\end{Lem}

\proof 
Let $s, t$ be two times such that $s\leq t \leq T$.
\begin{align*}
& \quad\,\,\esp\bigg[  \mathcal X_t(\nu)+\indic_{t \geq \tau_l} K_{\tau_l} U( X^{\nu,0}_{\tau_l}(l)(1-\nu_{\tau_l}(l)\gamma_{\tau_l}) )                |  \F_{{\tau_l} \wedge s} \bigg]  \\ 
&\, \, =   \esp\bigg[  \mathcal X_t(\nu)+\indic_{s< \tau_l \leq t} K_{\tau_l} U( X^{\nu,0}_{\tau_l}(l)(1-\nu_{\tau_l}(l)\gamma_{\tau_l}) )                |  \F_{{\tau_l} \wedge s} \bigg] 
+ \indic_{s \geq \tau_l} K_{\tau_l} U( X^{\nu,0}_{\tau_l}(l)(1-\nu_{\tau_l}(l)\gamma_{\tau_l}) )
\end{align*}
We make explicit the conditional expectation :
\begin{align}\label{surmgle}
&  \quad \,\,\,\esp\bigg[  \mathcal X_t(\nu)+\indic_{s< \tau_l \leq t} K_{\tau_l} U( X^{\nu,0}_{\tau_l}(l)(1-\nu_{\tau_l}(l)\gamma_{\tau_l}))                 |  \F_{{\tau_l} \wedge s} \bigg] \nonumber \\
& \, \,=   \esp\bigg[   
\esssup_{\pi^{0}\in \cA_l^{0,\delta}(t,\nu)}
\esp\big[\indic_{T<\tau_l} p_T(l) U( X_T^0(l)) + \indic_{ t< \tau_l \leq T}   K_{\tau_l} U( X^{0}_{\tau_l}(l)(1-\pi^{0}_{\tau_l}(l)\gamma_{\tau_l})   )     |  \F_{{\tau_l} \wedge t} \big] \nonumber \\
&\, \, \, \, \hspace*{3cm}  +\indic_{s< \tau_l \leq t} K_{\tau_l} U( X^{\nu,0}_{\tau_l}(l)(1-\nu_{\tau_l}(l)\gamma_{\tau_l}) )                |  \F_{{\tau_l} \wedge s} \bigg] \nonumber  \\
& \leq \esssup_{\pi^{0}\in \cA_l^{0,\delta}(s,\nu)}
\esp\big[\indic_{T<\tau_l} p_T(l) U( X_T^0(l)) + \indic_{ s< \tau_l \leq T}   K_{\tau_l} U( X^{0}_{\tau_l}(l)(1-\pi^{0}_{\tau_l}(l)\gamma_{\tau_l})   )     |  \F_{{\tau_l} \wedge s} \big]
\end{align}
the last inequality following from the fact that in the last esssup, the optimal control is taken from the date $s\leq t$.  Thus 
\begin{align*}
\esp\bigg[  \mathcal X_t(\nu)+\indic_{t \geq \tau_l} K_{\tau_l} U( X^{\nu,0}_{\tau_l}(l)(1-\nu_{\tau_l}(l)\gamma_{\tau_l})  )               |  \F_{{\tau_l} \wedge s} \bigg]   \leq     \mathcal X_s(\mu)
+ \indic_{s \geq \tau_l} K_{\tau_l} U( X^{\nu,0}_{\tau_l}(l)(1-\nu_{\tau_l}(l)\gamma_{\tau_l}) )
\end{align*}
and $(\xi_t^\nu)_{0 \leq  t \leq T}$
is an $\mathbb{F}^0$-supermartingale. It is an $\mathbb{F}^0$-martingale
if and only if the inequality (\ref{surmgle}) is an equality for all $t \in [0,T]$, meaning that $\nu$ is the optimal control on $[0,t ]$,  for all $t \leq T$. This characterizes the optimal strategy.
\finproof 

Remark that the $\mathbb{F}^0$-adapted process
\begin{align}\label{defYdynpro}
Y_t &:=  \frac{\mathcal X_t(\nu)}{U(X_t^{\nu,0}(l))}  \\
& =   \esssup_{\pi^{0} \in \cA_l^{0,\delta}(t,\nu)} \esp  \Big[ \indic_{T<\tau_l} p_T(l) \Big(\frac{X_T^{0}(l)}{X_t^{\nu,0}(l)}\Big)^p 
+   \indic_{ t< \tau_l \leq T}   K_{\tau_l}   \Big(\frac{X_{\tau_l}^{0}(l)}{X_t^{\nu,0}(l)}\Big)^p 
(1-\pi^{0}_{\tau_l}(l)\gamma_{\tau_l})  
 |  \F_{{\tau_l} \wedge t} \big]  \; \; \; \; \; 0 \leq t\leq T \nonumber
\end{align}
does not depend on $\nu$ $\in$ $\cA^{0,\delta}_l$, and is constant after $\tau_l$. We will give a characterization of the process $Y$ in terms of a backward stochastic differential equation (BSDE)
and of the optimal strategy. Before this, we give a characterization of  $\mathbb{F}^0$-martingale.

\begin{Lem}\label{Ftaumgle} 
Let $(M_t)_{t \in [ 0,T]}$ be an $\mathbb{F}^0$-martingale. Then there exists an $\mathbb{F}$-predictable process $\phi$  in $L_{loc}^2(W)$ such that $M_t=M_0+ \int_0^t \phi_s  \indic_{ s \leq \tau_l } dW_s, \, t\in [ 0,T]$.
\end{Lem}

\proof
We first prove that $(M_t)_{t \in [ 0,T]}$ is also an $\mathbb{F}$-martingale. Indeed, for $s\leq t \leq T$
$$M_s= \esp(M_t |\F_{{\tau_l} \wedge s} ) =\esp( \esp(M_t |\F_{{\tau_l}}) | \F_s  )= \esp( M_t | \F_s  ) $$
because $M_t$ is $\F_{\tau_l}$-measurable. Thus, by representation theorem for the $\mathbb{F}$-martingale, and since $(M_t)_{t\in [ 0,T]}$ is stopped at time $\tau_l$, there exists $\phi$ an $\mathbb{F}$-predictable process such that $M_t=M_0+ \int_0^t \phi_s  \indic_{ s \leq \tau_l } dW_s.$
\finproof

\noindent We are now ready to characterize the optimal strategy. Remark that $Y_t=  \frac{\mathcal X_t(\mu)}{U(X_t^{\nu,0}(l))} $ is positive on $\lbr 0,\tau_l \lbr$ (and zero after $\tau_l$)
 thus   $ Y \in L_{l}^+(\mathbb{F}^0)$   where
$$L_{l}^+(\mathbb{F}^0):= \{\tilde Y:   \, \,  \mathbb{F}^0{\mbox{-adapted processes s.t. }}  \tilde Y_t >0  \mbox{ for   } t \in \lbr 0,\tau_l \lbr  {\mbox{ and }}  \tilde Y_t =0  \mbox{ for   } t \in \lbr\tau_l,\infty\lbr    \,\} .$$

\begin{Thm}\label{thmbefore}
 The process $Y$ defined in (\ref{defYdynpro}) is  the smallest solution in  $L_{l}^+(\mathbb{F}^0)$ to the BSDE
\begin{equation}\label{BSDEY}
Y_t = \indic_{T<\tau_l} p_T(l)+  \indic_{ t < \tau_l \leq T  } K_{\tau_l}  \frac{ (1+\delta \gamma_{\tau_l})^p}{p}   +  \int_t^{T \wedge \tau_l}  f(\theta,Y_\theta,\phi_\theta  )  d \theta 
-  \int_t^{T \wedge \tau_l}  \phi_\theta  dW_\theta, \;\;\; t \in \lbr 0, T  \wedge \tau_l \rbr,
\end{equation}
for some  $\phi$ $\in$  $L_{loc}^2(W)$,  and where
\begin{equation} \label{driverf}
f(s,Y_s,\phi_s) = p \; \esssup_{\nu \in \cA^{0,\delta}_l, s.t. \,\,  \nu_{\tau_l}=-\delta } \Big[\big(\mu_s^{0} Y_s + \sigma_s^{0} \phi_s ) \nu_s 
- \frac{1-p}{2} Y_s |\nu_s\sigma_s^{0}|^2     \Big]
\end{equation}

\end{Thm}

\begin{Rem}
As in Theorem 4.2 in \cite{jp}, the optimal strategy before default is characterized through the optimization of the driver of a BSDE. However, the main difference relies in the fact that in our case, the driver has a jump at the default time $\tau_l$. Nevertheless, since the jump occurs (if it occurs) only at the terminal date of the BSDE, standard theory on BSDE still apply.
\end{Rem}

\proof By Lemma \ref{dynprosurmar}, for any $\nu$ $\in$ $\cA^{0,\delta}_l$
$$\xi_t^\nu= U(X_t^{\nu,0}(l))  Y_t+\indic_{t \geq \tau_l} K_{\tau_l} U( X^{\nu,0}_{\tau_l}(l)(1-\nu_{\tau_l}(l)\gamma_{\tau_l}) \, \, \,   0 \leq t \leq T  $$ is an  $\mathbb{F}^0$-supermartingale. In particular, by taking $\nu$ $=$ $0$, we see that the processes $(Y_t+  K_{\tau_l} \indic_{ t \geq \tau_l })_{0\leq t\leq T  }$, and thus $(Y_t)_{0\leq t\leq T  }$  are   $\mathbb{F}^0$-supermartingales.
By the Doob-Meyer decomposition and Lemma \ref{Ftaumgle}, there exists   $\phi$ $\in$ $L_{loc}^2(W)$, and a finite variation increasing $ \mathbb{F}^0$-predictable process $A$ such that:
\begin{equation}\label{dynY}
 dY_t  = \phi_t   dW_t - dA_t, \;\;\;  t\in \lbr 0,  T \wedge \tau_l \rbr .
\end{equation}
From It\^o's formula, 
we deduce that the finite variation process in the decomposition of the  $\mathbb{F}^0$-supermartingale $\xi^\nu$,
$\nu$ $\in$ $\cA^{0,\delta}_l$, is given by $-A^\nu$ with
\begin{equation*}
dA_t^\nu =(X_t^{\nu,0}(l))^p \Big\{  \frac{1}{p} dA_t -  (\mu_t^{0}Y_t + \sigma_t^{0}\phi_t \indic_{t \leq \tau_l} ) \nu_t dt - \frac{1-p}{2} Y_t |\nu_t\sigma_t^{0}|^2 dt
- K_t \frac{(1-\nu_t\gamma_t)^p}{p}   d\indic_{ t \geq \tau_l }  \Big\}.
\end{equation*}
$A^\nu$  is nondecreasing and the martingale property of $\xi^{\hat\pi^{0}}$ implies that 
\begin{align*}
  dA_t &=  p \Big[  (\mu_t^{0}Y_t + \sigma_t^{0}\phi_t \indic_{t \leq \tau_l} ) ) \hat\pi_t^{0}  dt - \frac{1-p}{2} Y_t |\hat\pi_t^{0}\sigma_t^{0}|^2 dt 
+ K_t \frac{(1-\hat\pi_t^{0}\gamma_t)^p}{p}  d\indic_{ t \geq \tau_l }    \Big]  \end{align*}
and $$A_t=A_0+  p \; \esssup_{\nu \in \cA^{0,\delta}_l} \Big[ \int_0^t\big((\mu_s^{0} Y_s + \sigma_s^{0} \phi_s ) \nu_s 
- \frac{1-p}{2} Y_s |\nu_s\sigma_s^{0}|^2\big)     ds  +  \indic_{ t \geq \tau_l } K_{\tau_l}  \frac{ (1-\nu_{\tau_l} \gamma_{\tau_l})^p}{p} \Big] .$$
Maximizing at  $\tau_l$ leads to $\nu_{\tau_l}=-\delta$ (see Proposition \ref{cvstrategie}) and 
$$A_t=A_0+  p \; \esssup_{\nu \in \cA^{0,\delta}_l s.t. \,\,  \nu_{\tau_l}=-\delta } \Big[ \int_0^t\big(\mu_s^{0} Y_s + \sigma_s^{0} \phi_s ) \nu_s 
- \frac{1-p}{2} Y_s |\nu_s\sigma_s^{0}|^2\big)     ds \Big] +  \indic_{ t \geq \tau_l } K_{\tau_l}  \frac{ (1+\delta  \gamma_{\tau_l})^p}{p}  .$$
Furthermore,  $Y_T= \indic_{T<\tau_l} p_T(l)$ and $(Y_t)_{0\leq t\leq T  }$ is constant after $\tau_l$ , thus $(Y,\phi)$ solves the BSDE (\ref{BSDEY}). Note that $Y$ is not a continuous process, it may jump at time $\tau_l$.

We now prove that $Y$ is  smallest solution in  the $L_{l}^+(\mathbb{F}^0)$ to the BSDE
(\ref{BSDEY}). Let   $\tilde Y \in L_l^+(\mathbb{F}^0)$ be another solution, and  we define the family of
nonnegative $\mathbb{F}^0$-adapted processes $\tilde\xi^\nu(\tilde Y)$, $\nu \in \cA^{0,\delta}_l$,
 as 
$$\xi_t^\nu(\tilde Y)= U(X_t^{\nu,0}(l)) \tilde Y_t+\indic_{t \geq \tau_l} K_{\tau_l} U( X^{\nu,0}_{\tau_l}(l)(1-\nu_{\tau_l}(l)\gamma_{\tau_l})), \, \,  \, \,  \, \, 
t\in [0,T] .$$
By similar calculations as above,  $d\xi_t^\nu(\tilde Y)=d\tilde M_t^\nu - d\tilde A_t^\nu$, where
 $\tilde A^\nu$ is a nondecreasing $\mathbb{F}^0$-adapted process, and $\tilde M^\nu$ is a $\mathbb{F}^0$- local martingale.  By Fatou's lemma,   this implies that the  process $\xi^\nu(\tilde Y)$ is a 
 $\mathbb{F}^0$-supermartingale, for any $\nu$ $\in$ $\cA^{0,\delta}_l $.  Since  $\tilde Y_T = \indic_{T<\tau_l} p_T(l)$, we deduce that for all $\nu$ $\in$ $\cA^{0,\delta}_l$
 $$\esp \Big[ U(X_T^{\nu,0}) \indic_{T<\tau_l} p_T(l) + 
\indic_{t \geq \tau_l} K_{\tau_l} U( X^{\nu,0}_{\tau_l}(l)(1-\nu_{\tau_l}(l)\gamma_{\tau_l}) 
\Big| \F^0_t \Big]  \leq
 U(X_t^{\nu,0}) \tilde Y_t,  \, \,  \, \,  \, \, 
t\in [0,T]. $$
 Since
 $p$ $>$ $0$, $U(X_t^{\nu,0})$ is positive.  By dividing the above inequalities by  $U(X_t^{\nu,0})$,  we deduce 
by definition of $Y$ (see (\ref{defYdynpro})), and arbitrariness of $\nu$ $\in$ $\cA^{0,\delta}_l$,  that   $Y_t$ $\leq$  $\tilde Y_t$, $0\leq t\leq T$. 
This shows that $Y$ is the smallest  solution to the BSDE (\ref{BSDEY}). 
\finproof

For optimizing \eqref{defYdynpro} via the BSDE (\ref{BSDEY}), a naive approach will consist in optimizing $\pi^0$ at time $\tau_l$, leading to an  $\pi^{0}_{\tau_l}=-\delta$, and then optimizing for $s < \tau_l$ the driver \[ f^0(s,Y^0_s,\phi^0_s) =\esssup_{\nu_s\geq-\delta}p[ \big(\mu_s^{0} Y^0_s + \sigma_s^{0} \phi^0_s ) \nu_s 
- \frac{1-p}{2} Y^0_s |\nu_s\sigma_s^{0}|^2]\] where    $Y^0$ is solution to the BSDE
$$Y^0_t = \indic_{T<\tau_l} p_T(l)  +  \indic_{ t < \tau_l \leq T  } K_{\tau_l}  \frac{ (1+\delta \gamma_{\tau_l})^p}{p} +\int_t^{T \wedge \tau_l}  f^0(\theta,Y_\theta^0,\phi_\theta^0  ) d\theta 
-  \int_t^{T \wedge \tau_l}  \phi^0_\theta  dW_\theta, \;\;\;  t\in \lbr 0,  T \wedge \tau_l \rbr ,$$
leading to the optimal portfolio $\hat \pi_s^0$.
Thus, the natural candidate to be the optimal strategy before default is   \begin{equation}\label{Equ:pinp}\pi^{\mathrm{np}}:= \indic_{\lbr0, \tau_l \lbr } \hat \pi^0 - \delta   \indic_{\tau_l } ,\end{equation} but unfortunately $\pi^{\mathrm{np}}$  is not a predictable process. Nevertheless, we will prove the existence of a sequence of predictable strategies in $\cA_l^{0,\delta}$ such that the corresponding value function tends to the value function relative to this non predictable strategy.
To do this, for any strategy $\pi^{0}\in\cA_l^{0,\delta}$, we recall the corresponding value function  of the before default global optimization problem
\begin{equation} \label{vbd}
\mathcal X_0(\pi^{0})=\esp\bigg[\indic_{T<\tau_l} p_T(l) U( X_T^{\pi^{0}}(l)) + \indic_{T\geq \tau_l}   K_{\tau_l} U( X^{\pi^{0}}_{\tau_l}(l)(1-\pi^{0}_{\tau_l}(l)\gamma_{\tau_l}) ) \bigg].
\end{equation}
Note that (\ref{vbd}) can also be defined for a strategy $\pi$  that is predictable only on $\lbr 0, \tau_l \lbr $ (and not necessarly on $\lbr 0, \tau_l \rbr$).
Using Proposition \ref{cvstrategie}, we have the following result:

\begin{Pro}\label{cvversoptimal}
The  strategies $ ( \pi^0_n =\indic_{\lbr 0, \tau_n\rbr} \hat \pi^0 - \indic_{\rbr \tau_n, \tau_l\rbr} \delta )$         
where $\hat \pi^0$ is the optimal process for the driver of the BSDE
$$Y^0_t = \indic_{T<\tau_l} p_T(l) +   \indic_{ t < \tau_l \leq T  } K_{\tau_l}  \frac{ (1+\delta \gamma_{\tau_l})^p}{p} + \int_t^{T \wedge \tau_l}  f^0(\theta,Y_\theta^0,\phi_\theta^0  ) d\theta 
-  \int_t^{T \wedge \tau_l}  \phi^0_\theta  dW_\theta, \;\;\; t \in \lbr 0, T  \wedge \tau_l \lbr,$$
$$ f^0(s,y,\phi) =p \, \sup_{\nu\geq-\delta} \Big[ \big(\mu_s^{0} y + \sigma_s^{0} \phi ) \nu - \frac{1-p}{2} y |\nu\sigma_s^{0}|^2\Big] $$
are in $\cA^{0,\delta}_l $  and  satisfy \, \, \, \, \, \,
$$\lim_{n \rightarrow + \infty}  \mathcal X_0(\pi^{0}_n ) =V(l)=  \esp\bigg[\indic_{T<\tau_l} p_T(l) U( X_T^{\hat \pi^{0}}(l)) + \indic_{T\geq \tau_l}   K_{\tau_l} U( X^{\hat \pi^{0}}_{\tau_l}(l)(1+\delta \gamma_{\tau_l}) ) \bigg]  .$$

\end{Pro}

\proof
Let $(\tau_n)_{n \in \mathbb{N}^*}$ be an increasing sequence of $\mathbb{F}$-predictable stopping times that converge to $\tau_l$. For any integer $n\geq 1$, the strategy $ \pi^0_n :=\indic_{\lbr 0, \tau_n \rbr} \hat \pi^0 - \indic_{\rbr \tau_n, \tau_l\rbr} \delta $  is in $\cA^{0,\delta}_l $  and $\pi^0_n$ converges to the non-predictable optimal strategy $\pi^{\mathrm{np}}$ defined in \eqref{Equ:pinp} when $n\rightarrow\infty$. Moreover, for any $n\in\mathbb N_*$, $\mathcal X_0( \pi^0_{n}) \leq \mathcal X_0( \pi^{\mathrm{np}})$ and by Proposition \ref{cvstrategie}
$$  \mathcal X_0( \pi^{\mathrm{np}}) \geq  \lim_{n \rightarrow + \infty} \mathcal X_0( \pi^0_{n}) \geq   V (\hat \pi^0 ).$$
But the proof of Proposition \ref{cvstrategie} still holds if we change the value at time $\tau_l$ of the portfolio $\pi^0$, thus the converse inequality $\mathcal X_0( \pi^{\mathrm{np}}) \leq  \lim_{n \rightarrow + \infty} \mathcal X_0( \pi^0_{n}) $ holds and 
\begin{align*}
\esp\bigg[\indic_{T<\tau_l} p_T(l) U( X_T^{\hat \pi^{0}}(l)) + \indic_{T\geq \tau_l}   K_{\tau_l} U( X^{\hat \pi^{0}}_{\tau_l}(l)(1+\delta \gamma_{\tau_l}) ) \bigg]&=  \mathcal X_0( \pi^{\mathrm{np}})= \mathcal X_0( \indic_{\lbr 0, \tau_l \lbr } \hat \pi^0 - \delta   \indic_{\tau_l }) \\
&=  \lim_{n \rightarrow + \infty} \mathcal X_0( \indic_{\lbr 0, \tau_n \rbr} \hat \pi^0 - \indic_{\rbr \tau_n, \tau_l \rbr} \delta) \\
 &=  \lim_{n \rightarrow + \infty} \mathcal X_0(\pi^0_{n})
\end{align*}
\finproof

\section{Numerical illustrations}

We now illustrate our previous results by explicit models and we aim to compare the  optimization results for an insider and a standard investor. We recall that all investors start with an initial wealth $X_0$. For the purpose of comparison, we choose a similar model with the one studied in \cite{jp}. More precisely, we let the parameters $\mu^{0}$, $\sigma^{0}$, $\gamma$  to be constant, and $\mu^1(\theta)$, $\sigma^1(\theta)$ are deterministic functions of $\theta$ given by \begin{eqnarray}
\mu^1(\theta) \; = \; \mu^{0}\frac{\theta}{T}, & &
\sigma^1(\theta) \; = \; \sigma^{0}(2-\frac{\theta}{T}), \;\;\; \theta \in [0,T],
\end{eqnarray}
which means that the ratio of the after-default and before-default for the return rate of the asset is smaller than $1$ and for the volatility is larger than $1$.  Moreover, these ratios increase or decrease linearly 
with the default time respectively:  the after-default rate of return drops to zero, when the default time occurs near the initial date, and 
converges  to the  before-default rate of return, when  the default time occurs near the finite investment horizon.  
For the volatility,
this ratio  converges to the double (resp. initial) value of the before-default  volatility, when the default time goes to the initial (resp. terminal horizon) time. Moreover, in order to satisfy the hypothesis in the simulation part of  \cite{jp}, we have to assume that the default barrier $L$ has no atoms (to ensure the density hypothesis, see Remark \ref{Rklebesgue})  and that  $L$ is independent of the filtration $\bF$ (so that the default density is a deterministic function). In this case,  $p_T(L)=1$. 

Consider the CRRA utility $U(x)=\frac{x^p}{p}$, $0<p<1$, the after-default value function is given from \eqref{V1 CRRA} by
\[V^1_{\tau_l}(x)=K_{\tau_l}U(x)\]
where
\[K_{\tau_l}=\left(\esp[Z_T(\tau_l)^{\frac{p}{p-1}}]\right)^{1-p}=\exp\left(\frac 12\frac{p}{1-p}\Big(\frac{\mu^1(\tau_l)}{\sigma^1(\tau_l)}\Big)^2(\tau_l\vee T-\tau_l)\right)\]
Furthermore, the solution of the before-default optimization problem is given by 
\[V_0(l)=Y_0U(X_0)\]
where $Y$ is the solution of the BSDE \eqref{BSDEY} when letting $\phi=0$, i.e. 
\begin{equation}\label{numerical Y}Y_t=\indic_{T<\tau_l}+\indic_{T\geq\tau_l} K_{\tau_l}\frac{(1+\delta\gamma)^p}{p}+\int_t^{\tau_l\wedge T}f(\theta,Y_\theta)d\theta\end{equation}
where 
\[f(t,y)=p \esssup_{\nu\in\mathcal A_l^{0,\delta}, \nu_{\tau_l}=-\delta}\{\mu^0\nu_t-\frac{1-p}{2}(\nu_t\sigma^0)^2\} y. \]
We notice that in the case where the default time $\tau_l$ occurs after the maturity $T$, the optimal strategy coincides with the classical Merton strategy with constraint $\pi\in[-\delta, \frac{1}{\gamma}[$.
In the case where $\tau_l$ occurs before $T$, the process $Y$ is stopped at $\tau_l$, with the terminal value depending on the quantity $K_{\tau_l}$, and the short-selling strategy $-\delta$ at $\tau_l$. We use an iterative algorithm to solve the equation \eqref{numerical Y}. 

The following results are based on the model parameters described below: $\mu^0=0.03$, $\sigma^0=0.2$, $T=1$, the risk aversion parameter $p=0.8$. For the standard investor, we use the deterministic model as in \cite{jp} letting $\proba(\tau>t)=e^{-\lambda t}$, $\lambda>0$ so that the density function is $\alpha(\theta)=\lambda e^{-\lambda\theta}$ for all $\theta\geq 0$.    

Figure \ref{graphe_merton} compares the optimal value function for insider, investor and Merton strategy. We  fix the short-selling constraint $\delta =0.5$, the loss given default $\gamma=0.2$ and the default intensity $\lambda= 0.3$. This corresponds to a relatively high risk of default. At the default time, the value function suffers a brutal loss for all the three strategies. The insider outperforms the other two strategies before and after the default occurs. Before the default, the value function for the standard investor is smaller than the Merton one because the latter does not consider at all  the potential default risk. However, when the default occurs, the investor outperforms the Merton strategy since the default risk is taken into account from the beginning. 

\begin{figure}[h]\caption{\small{Value function insider vs investor and Merton: $p=0.2$, $\lambda=0.3$ and $\delta =0.5$.}}\label{graphe_merton} 
\begin{center}
\epsfig{file=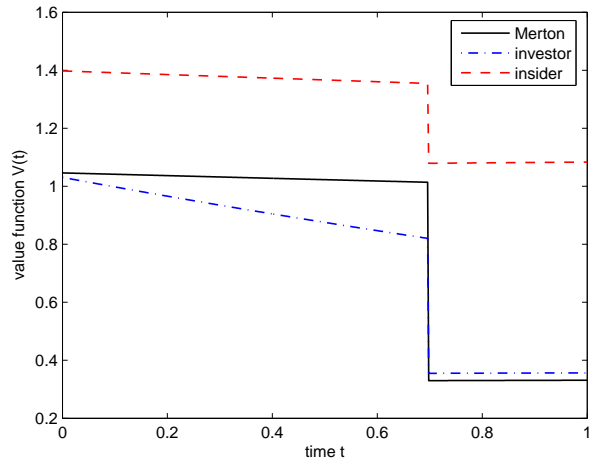,height=9cm,angle=270} 
\end{center}

\end{figure}

In Figure \ref{rolegamma1}, we  fix a smaller  short-selling constraint $\delta=0.1$ and we study the impact of the loss given default $\gamma$. We observe that the value functions for both insider and investor are increasing with respect to the loss value $\gamma$.  It is interesting to emphasize this phenomenon as a consequence of the short-selling where both insider and the investor will profit the default event and obtain a larger value function.  More precisely, in the left-hand figure, we consider a relatively low default risk (with the default intensity $\lambda= 0.1$) and we observe that the gain of the insider with respect to the investor remains stable in time and also for different values of $\gamma$. Whereas in the right-hand figure  with a higher default risk ($\lambda= 0.3$), the impact of $\gamma$ is more important for the investor:  before the default the value function of the investor increases more significantly as $\gamma$ increases, and there is no longer a drop at the default time.    
This can be explained by the fact that the investor has no limit for the short-selling strategy.  

\begin{center}
\begin{figure}[h]\caption{Value function insider vs investor:   $\delta =0.1$ }\label{rolegamma1}
$$\begin{array}{ccc}
\includegraphics[width=6cm,angle=270]{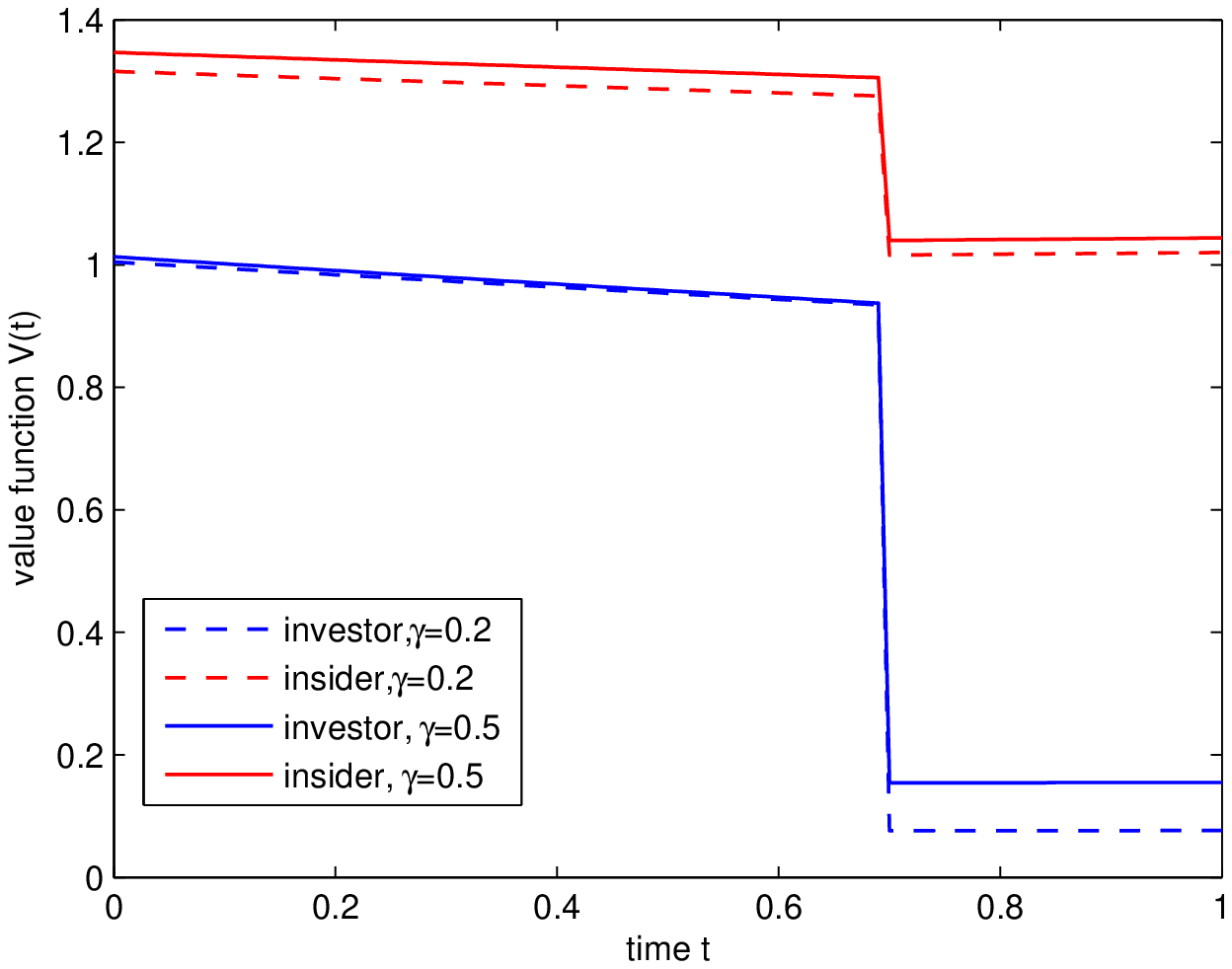} &
\includegraphics[width=6cm,angle=270]{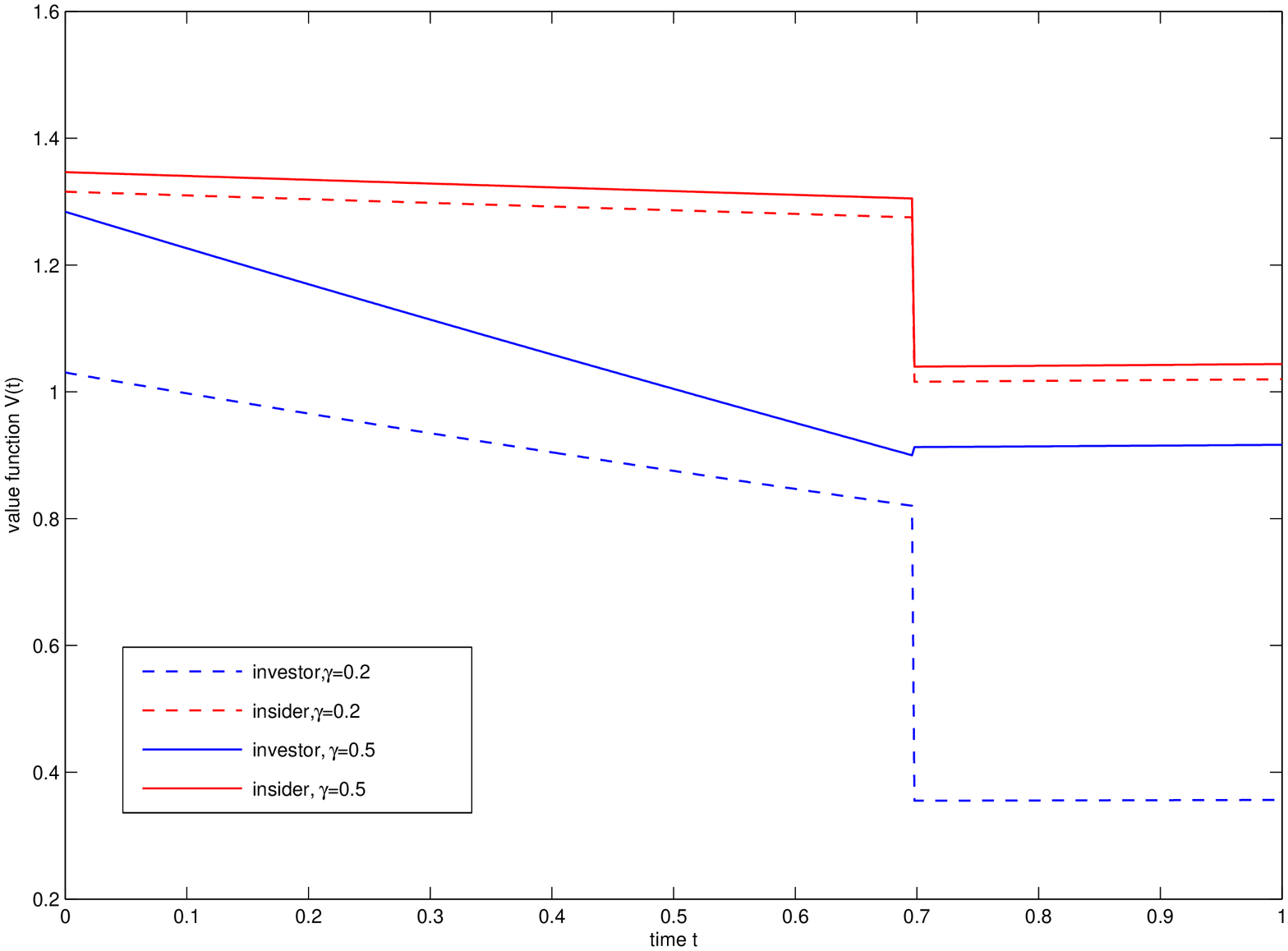} \\
\mbox{$\lambda=0.1$ } & \mbox{$\lambda=0.3$ }
\end{array}$$
\end{figure}
\end{center}

Figure \ref{rolegamma2} shows what may happen in a extreme situation for the default and loss risks. 
For extreme parameters of default intensity $\lambda=0.5 $ and loss value $\gamma=0.5$, we  remark that  the investor may outperform the insider, see the left-hand figure. In this situation, the investor bet on the occurrence of the default before $T$ and short sell a big amount of the risky asset. The insider value function at time 0 is higher for the investor, and then decreases  rapidly until the default time : indeed, as the time goes on and the default has not occurred yet, the risk that the bet turns out to be wrong increases, and thus the value function decreases. At the default time, the investor makes profit of the short selling strategy, the value function being almost doubled. The right-hand figure illustrates  the wealth process of the investor on a given scenario in which the default occurs after $T$ despite    extreme parameters for the default and loss risks ($\lambda=0.5 $ and  $\gamma=0.5$). We observe large losses, which are  induced by the wrong bet on default and extreme short selling positions.

\begin{center}
\begin{figure}[h]\caption{Extreme scenarios of default :   $\lambda=0.5$, $\delta =0.1$ }\label{rolegamma2}
$$\begin{array}{ccc}
\includegraphics[width=6cm,angle=270]{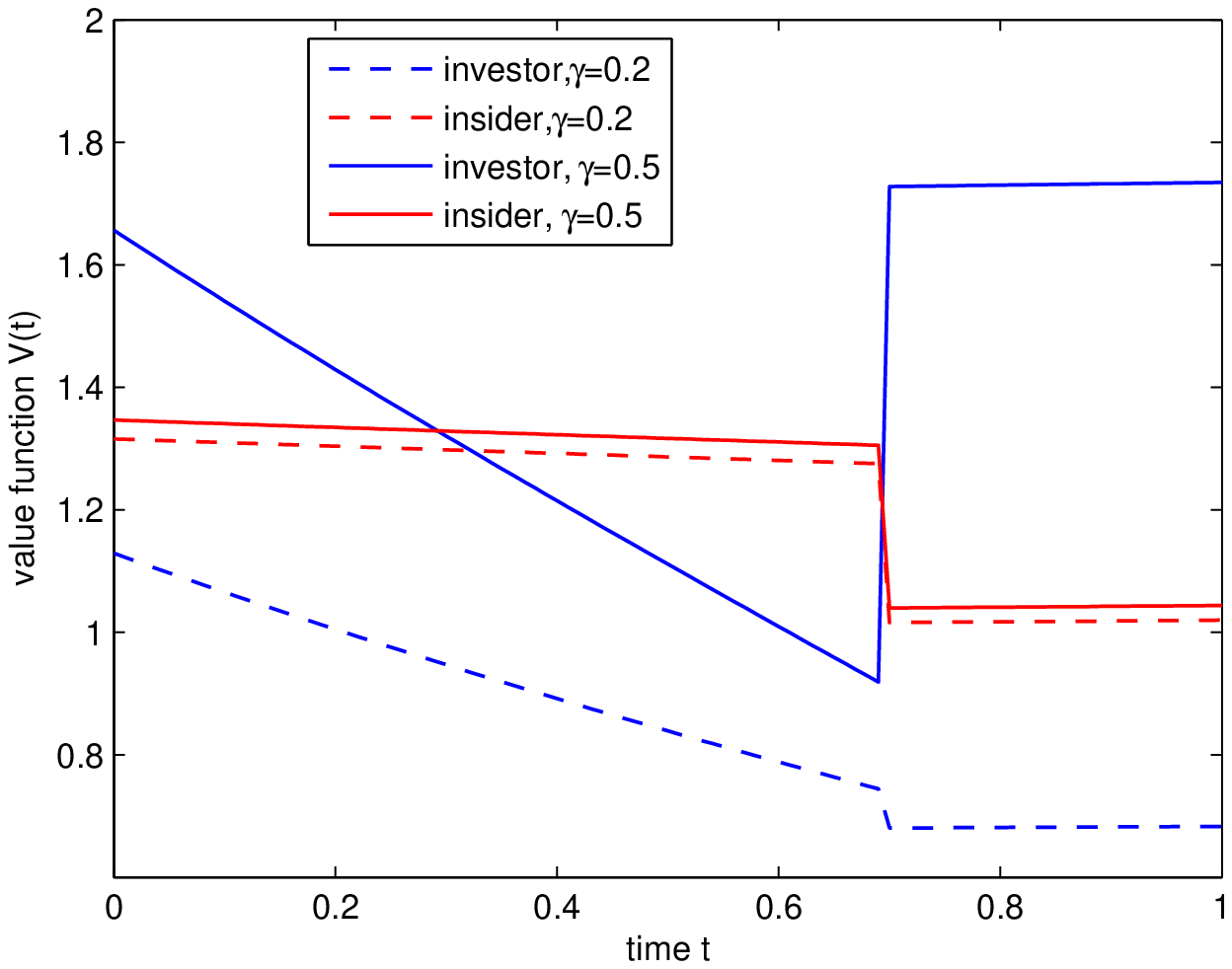} &
\includegraphics[width=6cm,angle=270]{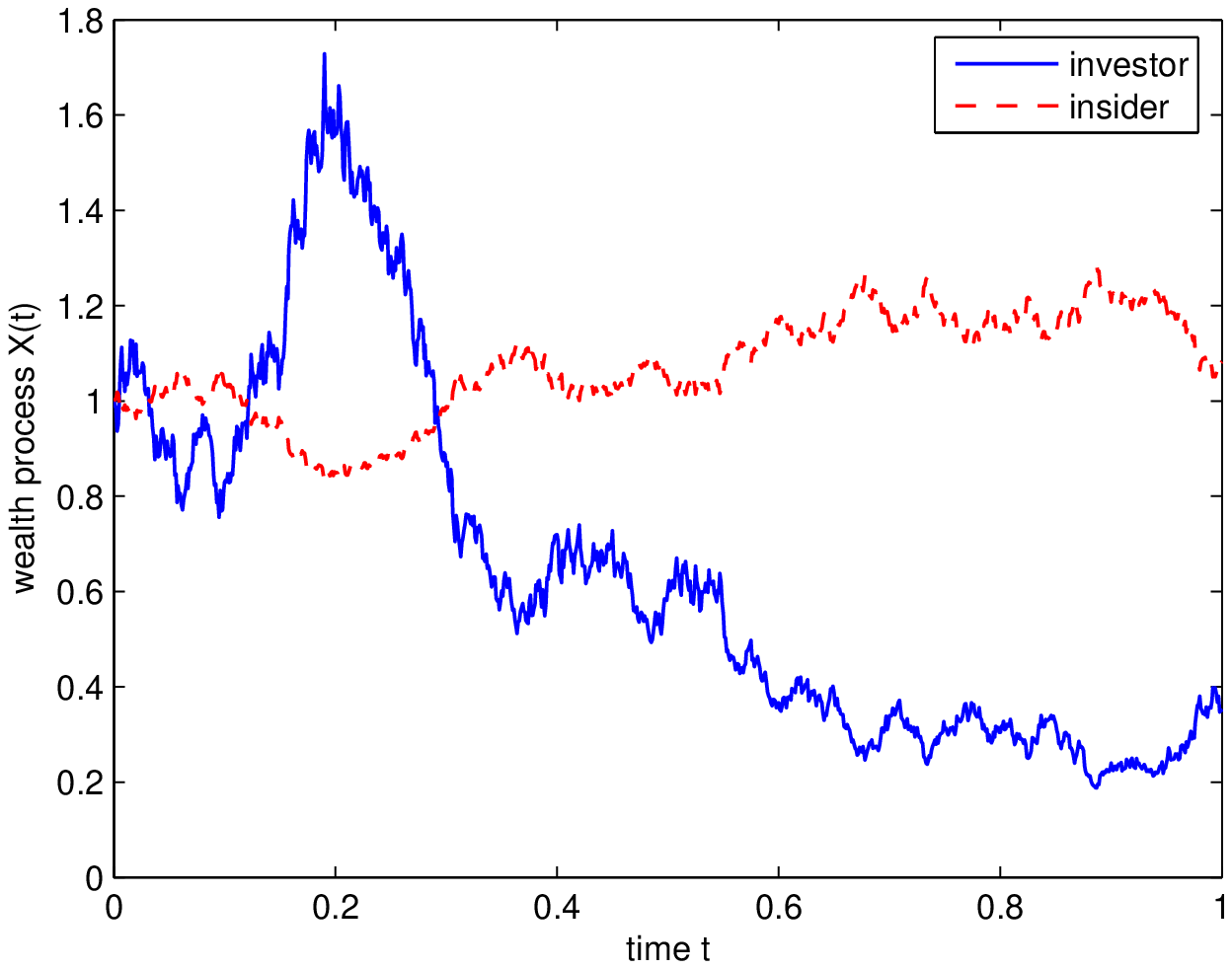} \\
\mbox{\small{Value function insider vs investor} } & \mbox{\small{Wealth process on one trajectory with $\tau>T$, $\gamma=0.5$} }
\end{array}$$
\end{figure}
\end{center}

Finally, we study the role of the short-selling limit for different values of $\delta$ in Figure \ref{role delta 1}. Not surprisingly, the gain of the insider is an increasing function of $\delta$. 
\begin{figure}[h]\caption{\small{The impact of the short-selling constraint,  $\lambda=0.3$ and $\gamma =0.5$.}}\label{role delta 1} 
\begin{center}
\epsfig{file=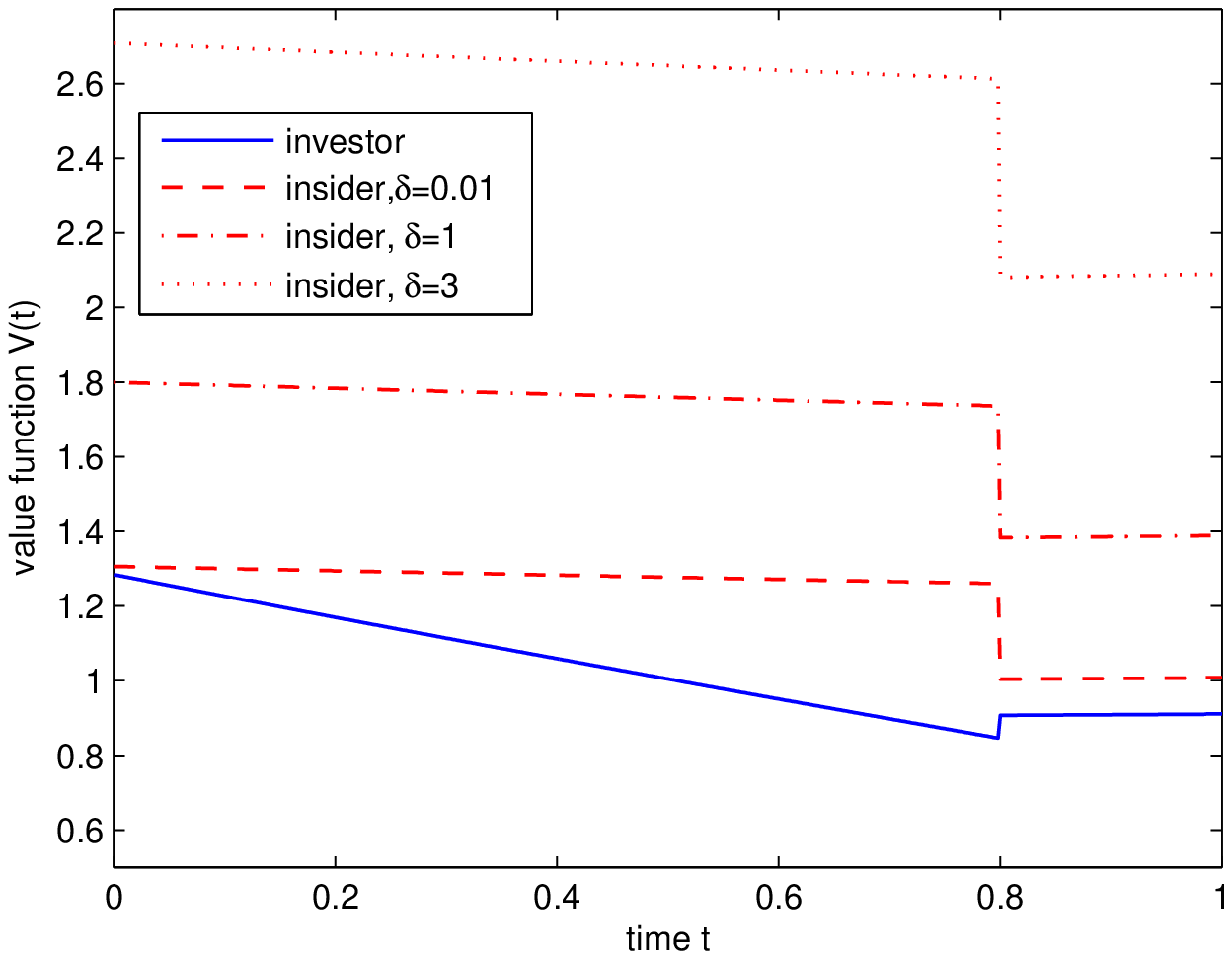,height=9cm,angle=270} 
\end{center}
\end{figure}

\end{document}